\documentclass[conference]{IEEEtran}

\usepackage{enumitem}

\usepackage[dvipsnames]{xcolor}

\usepackage[most]{tcolorbox}
\tcbset{on line, 
        boxsep=4pt, left=0pt,right=0pt,top=0pt,bottom=0pt,
        colframe=white,colback=gray!10,  
        highlight math style={enhanced}
        }
        
\usepackage[noadjust]{cite}
\usepackage{filecontents}

\let\svtikzpicture\tikzpicture
\def\tikzpicture{\noindent\svtikzpicture}

\usepackage{amsbsy}
\usepackage{bm}
\usepackage{fixmath}
\usepackage{silence}
\usepackage{ctable}
\usepackage{pst-node,pst-plot}
\usetikzlibrary{shapes.geometric}
\usepackage{mathdots}
\usetikzlibrary{math}

\usepackage{upgreek}

\usepackage{autobreak}
\allowdisplaybreaks

\usepackage[caption=false]{subfig}
\usepackage[thinc]{esdiff}
\usepackage{commath,amsmath,amssymb,amsfonts}
\usepackage{mathtools}
\usepackage{amsthm}
\usepackage{algorithmic}
\usepackage{pgfplots} 
\pgfplotsset{compat=1.15}
\usepackage{graphicx}
\usepackage{pgfgantt}
\usepackage{pdflscape}
\usepackage{pst-plot}
\usepackage{xfrac}
\usepackage{colortbl}
\usepackage{cancel}
\usepgfplotslibrary{fillbetween}
\usepackage{amssymb,bm}
\usepackage{float}
\usepackage{amsmath}
\usepackage{multirow}

\usepackage{mathrsfs}
\usepackage{bbm}

\newcommand{\floor}[1]{\left \lfloor #1 \right \rfloor}

\usepackage{cite} 

\usepackage{comment}

\usepackage{autobreak}
\allowdisplaybreaks

\def \D{\mathcal{D}}

\def \M{\mathcal{M}}

\def \P{\mathcal{P}}

\def \S{\mathcal{S}}

\def \Z{\mathcal{Z}}

\def \fA{\mathbf{A}}
\def \fB{\mathbf{B}}
\def \fU{\mathbf{U}}
\def \fV{\mathbf{V}}

\def \fc{\mathbf{c}}

\def \fv{\mathbf{v}}
\def \fV{\mathbf{V}}
\def \fx{\mathbf{x}}
\def \fy{\mathbf{y}}

\def \fU{\mathbf{U}}

\def \fP{\mathbf{P}}

\def \fV{\mathbf{V}}
\def \fX{\mathbf{X}}
\def \fY{\mathbf{Y}}

\def \f0{\mathbf{0}}

\definecolor{blau_1a}{RGB}{93,133,195}
\definecolor{blau_2a}{RGB}{0,156,218}
\definecolor{gruen_3a}{RGB}{80,182,149}
\definecolor{gruen_4a}{RGB}{175,204,80}
\definecolor{gruen_5a}{RGB}{221,223,72}
\definecolor{orange_6a}{RGB}{255,224,92}
\definecolor{orange_7a}{RGB}{248,186,60}
\definecolor{rot_8a}{RGB}{238,122,52}
\definecolor{rot_9a}{RGB}{233,80,62}
\definecolor{lila_10a}{RGB}{201,48,142}
\definecolor{lila_11a}{RGB}{128,69,151}

\definecolor{blau_1b}{RGB}{0,90,169}
\definecolor{blau_2b}{RGB}{0,131,204}
\definecolor{gruen_3b}{RGB}{0,157,129}
\definecolor{gruen_4b}{RGB}{153,192,0}
\definecolor{gruen_5b}{RGB}{201,212,0}
\definecolor{orange_6b}{RGB}{253,202,0}
\definecolor{orange_7b}{RGB}{245,163,0}
\definecolor{rot_8b}{RGB}{236,101,0}
\definecolor{rot_9b}{RGB}{230,0,26}
\definecolor{lila_10b}{RGB}{166,0,132}
\definecolor{lila_11b}{RGB}{114,16,133}

\definecolor{mycolor1}{rgb}{0.0, 0.18, 0.39}
\definecolor{mycolor2}{RGB}{87,108,67}
\definecolor{mycolor3}{RGB}{8,133,161}
\definecolor{mycolor4}{RGB}{80,91,161}
\definecolor{mycolor5}{RGB}{98,122,157}
\definecolor{mycolor6}{RGB}{255,163,67}
\definecolor{mycolor7}{RGB}{152,205,225}
\definecolor{mycolor8}{RGB}{242,204,48}
\definecolor{mycolor9}{rgb}{0,.5,0}
\definecolor{mycolor10}{rgb}{.59,.44,.09}
\definecolor{mycolor11}{RGB}{231,199,31} 
\definecolor{mycolor12}{RGB}{8,133,161} 
\definecolor{mycolor13}{RGB}{157,188,64} 
\definecolor{mycolor14}{RGB}{194,150,130} 
\definecolor{mycolor15}{RGB}{98,122,157} 
\definecolor{mycolor16}{RGB}{160,160,160} 
\definecolor{mycolor17}{RGB}{115,82,68} 
\definecolor{mycolor18}{RGB}{94,60,108} 
\definecolor{mycolor19}{RGB}{115,82,68} 
\definecolor{mycolor20}{RGB}{255,183,30} 

\usepackage{accents}

\theoremstyle{remark} \newtheorem{theorem}{Theorem}
\theoremstyle{remark} 
\theoremstyle{remark} \newtheorem{definition}{Definition}
\theoremstyle{remark} \newtheorem{remark}{Remark}
\theoremstyle{remark} 

\newtheorem{innercustomgeneric}{\customgenericname}
\providecommand{\customgenericname}{}
\newcommand{\newcustomtheorem}[2]{%
  \newenvironment{#1}[1]
  {%
   \renewcommand\customgenericname{#2}%
   \renewcommand\theinnercustomgeneric{##1}%
   \innercustomgeneric
  }
  {\endinnercustomgeneric}
}
\newcustomtheorem{customcorollary}{Corollary}
\newcustomtheorem{customlemma}{Lemma}
\newcustomtheorem{customproposition}{Proposition}

\usepackage{amssymb}

\usepackage{autobreak}
\allowdisplaybreaks

\usepackage{tikz}
\usetikzlibrary{arrows.meta}

\usepackage{blochsphere}

\usepackage{multicol, blindtext}

\usepackage[]{pgfplots}
\usepackage{pgfplotstable}
\usepgfplotslibrary{dateplot}
\pgfplotsset{/pgf/number format/use comma,compat=newest}

\usepackage{amsbsy}
\usepackage{bm}
\usepackage{fixmath}

\usepackage{silence}
\usepackage{ctable}

\usepackage{pst-node,pst-plot}

\usetikzlibrary{shapes.geometric}

\usepackage{mathdots}

\usetikzlibrary{math}

\definecolor{amber}{rgb}{1.0, 0.49, 0.0}
\definecolor{cadmiumgreen}{rgb}{0.0, 0.42, 0.24}

\newtheoremstyle{styleth}%
{3pt}
{3pt}
{}
{}
{\bfseries\color{mycolor3}}
{}
{.5em}
{}
\theoremstyle{styleth}

\newtheoremstyle{styledef}%
{3pt}
{3pt}
{}
{}
{\bfseries\color{cadmiumgreen}}
{}
{.5em}
{}
\theoremstyle{styledef}

\usepackage{stackengine}

\usepackage{feyn}
\usepackage{mathtools,multirow}

\def\fc{{\bf c}}
\def\fA{{\bf A}}
\def\fU{{\bf U}}
\def\fV{{\bf V}}
\def\fW{{\bf W}}
\def\fB{{\bf B}}
\def\fx{{\bf x}}
\def\fDelta{\boldsymbol{\mathbf{\Updelta}}}
\def\fSigma{\boldsymbol{\mathbf{\Sigma}}}

\usetikzlibrary{arrows, arrows.meta}
\usepackage{pgfplots}
\pgfplotsset{compat=1.15}
\usepgfplotslibrary{polar}

\usepackage[bottom]{footmisc}

\usepackage{tcolorbox}

\definecolor{codeBg}{rgb}{0.976, 0.949, 0.956}
\definecolor{codeColor}{rgb}{0.780, 0.145, 0.305}

\newtcbox{\bCode}{
    nobeforeafter,
    fontupper=\color{blau_2b!70!gray},
    colframe=codeBg,
    colback=gray!7,
    boxrule=0.1pt,
    arc=7pt,
    boxsep=0pt,
    left=3pt,
    right=3pt,
    top=3pt,
    bottom=4pt,
    tcbox raise base}

\usepackage{lipsum}   
\usepackage{framed}
\usepackage{tikz}
\usetikzlibrary{decorations.pathmorphing,calc}
\pgfmathsetseed{1} 
\pgfdeclarelayer{background}
\pgfsetlayers{background,main}

\tikzset{
  normal border/.style={mycolor2!4, decorate, 
     decoration={random steps, segment length=12mm, amplitude=2mm}},
  torn border/.style={orange!30!black!5, decorate, 
     decoration={random steps, segment length=.5cm, amplitude=1.7mm}}}

{\endMakeFramed}

\newlength{\PRLlen}
\newcommand*\PRLsep[1]{\settowidth{\PRLlen}{#1}\advance\PRLlen by -\textwidth\divide\PRLlen by -2\noindent\makebox[\the\PRLlen]{\resizebox{\the\PRLlen}{1.4pt}{$\blacktriangleleft$}}\raisebox{-.5ex}{#1}\makebox[\the\PRLlen]{\resizebox{\the\PRLlen}{1.4pt}{$\blacktriangleright$}}\bigskip}

\usepackage{tikz,tikz-3dplot}
\usetikzlibrary{calc,arrows.meta}

\DeclareMathOperator*{\argmax}{arg\,max}

\usepackage{nccmath}

\usepackage{tikz}
\usepackage{mathtools}

\def\nudge{.5}

\tikzset{axis/.style={ultra thick, Red!75!black, -latex, shorten <=-\nudge cm, shorten >=-2*\nudge cm}}
\tikzset{line/.style={thick,Green}}

\definecolor{MK_Three_Five}{RGB}{127,191,123}
\usepackage[draft=false]{hyperref}
\hypersetup{
	colorlinks=false				
	,linkbordercolor=RoyalBlue
	,citebordercolor= MK_Three_Five
	,urlbordercolor= cyan
}

\usepackage{cleveref}
\crefname{equation}{Eq}{} 
\crefrangelabelformat{equation}{(#3#1#4--#5#2#6)}

\allowdisplaybreaks

\font\myfont=cmr12 at 21.8pt
\font\myfontt=cmr12 at 19.84pt

\makeatletter
\patchcmd{\@maketitle}
  {\addvspace{0.5\baselineskip}\egroup}
  {\addvspace{-.5\baselineskip}\egroup}
  {}
  {}

\begin{document}

\title{{\myfont Identification over Affine Poisson Channels:} \\[-.1em] {\myfontt{Application to Molecular Mixture Communication Systems}}\vspace{-4mm}}
\author{\vspace{2mm} \fontsize{12}{12} \selectfont Mohammad Javad Salariseddigh\IEEEauthorrefmark{1}, Heinz Koeppl\IEEEauthorrefmark{2}, Holger Boche\IEEEauthorrefmark{3} and Vahid Jamali\IEEEauthorrefmark{1}
\vspace{.2mm}
\\
\fontsize{9}{9} \selectfont \IEEEauthorrefmark{1}Resilient Communication Systems Group, Technical University of Darmstadt (TUDa), \IEEEauthorrefmark{2}Self-Organizing Systems Lab, TUDa
\\
\fontsize{9}{9} \selectfont \IEEEauthorrefmark{3}Chair of Theoretical Information Technology, Technical University of Munich
}

\maketitle

\begin{abstract}
Identification capacity has been established as a relevant performance metric for various goal-/task-oriented applications, where the receiver may be interested in only a particular message that represents an event or a task. For example, in olfactory molecular communications (MCs), odors or pheromones, which are often a mixture of various molecule types, may signal nearby danger, food, or a mate. In this paper, we examine the identification capacity with deterministic encoder for the discrete affine Poisson channel which can be used to model MC systems with molecule counting receivers. We establish lower and upper bounds on the identification capacity in terms of features of the affinity matrix between the released molecules and receptors at the receiver. As a key finding, we show that even when the number of receptor types scales sub-linearly in the number of molecule types $N,$ the number of reliably identifiable messages can grow super-exponentially with the rank of the affinity matrix, $T,$ i.e., $\sim 2^{(T \log T)R},$ where $R$ denotes the coding rate. We further derive lower and upper bounds on $R,$ and show that the proposed capacity theorem includes several known results in the literature as its special cases.
\end{abstract}

\IEEEpeerreviewmaketitle

\vspace{-1mm}
\section{Introduction}
In the identification problem \cite{J85,Salariseddigh_IT,AD89}, appropriate communication strategies enable the receiver to reliably answer the following question: Whether a specific message (which is e.g., relevant for a specific task), is sent by the transmitter or not? In other words, in contrast to the message transmission problem of Shannon \cite{S48} where the decoder seeks to determine which message from the entire codebook was sent, here, the decoder is solely focused on a specific message and tries to identify whether or not that message was sent. Prior to the initiation of communication, the transmitter is unaware of the specific message that the receiver seeks for identification purposes.

In the identification problem for a discrete memoryless channel (DMC), when employing deterministic encoders, the codebook size exhibits an exponential growth with the codeword length \cite{Salariseddigh_IT} (mirroring the canonical behavior observed in \cite{S48}). However, under a randomized encoding scheme the identification capacity features a dramatically amplified behavior, namely, the requisite codebook size scales \emph{double exponentially} in the codeword length $N,$ i.e., $\sim 2^{2^{NR}}$ \cite{AD89}. In line with Shannon's findings \cite{S48}, Ahlswede developed a celebrated identification capacity theorem for a DMC and established that identification capacity is equal to its Shannon capacity \cite{AD89}. For the proof of achievability and converse, Ahlswede employed an intricate combination of tools from the the method of types and set theory. Recently, the work \cite{Colomer_2025} extends the identification problem for memoryless channels with finite output and arbitrary input alphabets. Identification problem has received a considerable attention for studying a variety of emerging applications in the context of post Shannon and semantic goal-oriented communications; see \cite{Salariseddigh23_BSC_Future_Internet,6G_PST} for more discussions. In particular, the identification problem may be adopted to event recognition which find applications in the context of 6G wireless networks \cite{6G_PST}, disease detection in molecular communications (MC) \cite{6G+} and discriminating olfactory percepts in natural olfaction \cite{Varshney16}.

Identification is particularly relevant for MC systems, where instead of transmitting a large amount of data, the objective is often to identify an event or communicate a task. For instance, in olfactory communications, pheromones and odorants are molecular mixtures (a message consists of combination of distinct types of molecules \cite{Buck05}) that are used for signaling danger, locating food, initiating mating, etc. In particular, the human olfactory system, with its hundreds of different olfactory receptors ($\sim 400$ types) can discriminate at a huge number of stimuli \cite{Bushdid14}. See \cite{Meister14,Meister15} for revisiting the mathematical approach used in \cite{Bushdid14} and refinements regarding the rate. That is, the purpose of communication in such scenarios may involve the \textit{recognition} of a specific type of secreted odor/pheromone/stimuli. Thereby, the identification problem \cite{Jamali22} may be adopted to natural olfaction. Here, we focus on the affine Poisson channel, where the Poisson statistics model the random molecule counting process and the affinity matrix, $\fA,$ in the affine model characterizes the connection between the number of activated receptors and the number/rate of released molecules of various types. We note that as its special case, the affine Poisson channel can model conventional time-slotted MCs, where each channel use corresponds to release of molecules of the same types, by choosing $\fA$ to be the identity matrix for inter-symbol interference (ISI)-free channel \cite{Salariseddigh-TMBMC} and $\fA$ assuming a Toeplitz structure for ISI channel \cite{Salariseddigh_OJCOMS_23}.
Furthermore, in addition to modeling molecule-counting receivers in MC systems, the Poisson channel has also been used to model photon-counting receivers in optical communication systems \cite{Salariseddigh-TMBMC}, and packet traffic arrival in conventional telecommunication networks \cite{Jain_86}. While the identification capacity has been studied in previous studies for ISI-free \cite{Salariseddigh-TMBMC} and ISI Poisson channels \cite{Salariseddigh_OJCOMS_23}, to the best of the authors’ knowledge, it has not been investigated for the general affinity Poisson channel, yet.

In this paper, we consider identification systems over the discrete affine Poisson channels (DAPC) employing deterministic encoder in the presence of average and peak molecule release rate constraints, and make the below contributions:
\vspace{-1mm}
\begin{itemize}[leftmargin=*]
    \item \textbf{\textcolor{blau_2b}{Generalized Poisson Model:}} Capacity analysis often studies an asymptotically large number of channel uses, whereby in conventional time-slotted MC systems, the number of transmitted symbols and the number of received observations are the same (or linearly related). In the considered affine Poisson channel, the number of molecule types in the molecular mixtures, denoted by $N,$ and the number of receptor types, denoted by $K,$ are analogous to the number of transmitted symbols and received observations, respectively. We formulate the capacity problem where $N$ and $K$ are scaled differently with respect to each other. Moreover, we show that reliable identification is possible even when $K$ grows sub-linearly in $N,$ i.e., $K = N^{\kappa}$ where $\kappa \in (0,1]$ is called the receptors rate in $N.$
    \item \textbf{\textcolor{blau_2b}{Codebook Scale:}} We establish that the codebook size of the DAPC for deterministic encoding scales super-exponentially in the rank of the affinity matrix, $T,$ i.e., $\sim 2^{(T \log T)R},$ where $R$ is the coding rate. Interestingly, this scaling behavior coincides with the scaling behavior as observed for conventional ISI and ISI-free Poisson channels, where the number of observations (i.e., $T \propto K \propto N$) scales linearly in the number of transmitted symbols, despite the fact that here the number of observations in our setup only grows sub-linearly in $N,$ i.e., $K = N^{\kappa}$ with $\kappa \in (0,1].$
    \item \textbf{\textcolor{blau_2b}{Capacity Bounds:}} We derive lower and upper bounds on $R$ for a general class of affinity matrices $\fA$ where the bounds are a function of the parameters of $\fA.$ In particular, these bounds depend on the dimension of $\fA$ (namely how the number of receptors $K$ scales with the number of molecule types $N$) and the sparsity of $\fA$ (i.e., the number of non-zero elements in each row of $\fA$). More specifically, assume that the number of non-zero entries in each row of $\fA,$ denoted by $F_k,$ scales at most sub-linearly in the rank of the affinity matrix, $T,$ i.e., $F_k = \mathcal{O}(T^l),$ for every $k \in \{1,2,\ldots,K\},$ where $l \in [0,1)$ is referred to as the non-zero rate across all rows. Then, the lower and upper bounds on the capacity, decreases and increases in $l,$ respectively.
    \item \textbf{\textcolor{blau_2b}{Proof Novelty:}}
    For the proposed lower bound, we show that there exists a suitable arrangement of spheres where the distance between the centers of the spheres is above a specific threshold. This packing incorporates the effect of two key features of the affinity matrix $\fA,$ namely, the non-zero rate across all rows $l,$ and the receptors rate, $\kappa.$
\end{itemize}
\vspace{-1mm}
\textbf{Notations:}
Alphabet sets are denoted by blackboard letters $\mathbb{X,Y},\ldots.$ The set difference is indicated by $\mathbb{X}\setminus\mathbb{Y}.$ Lowercase letters $x, y, \ldots$ represent constants and values of random variables (RVs), while uppercase letters $X, Y, \ldots$ represent RVs. Row vectors are shown by lowercase bold symbols $\fx$ and $\fy.$ The operator $\dag$ shows the transpose of a vector. The symbol $\langle.,.\rangle$ shows the inner product. $\fX_{P \times G}$ denotes that matrix $\fX$ has $P$ rows and $G$ columns with $\fX_p$ and $\text{rank}(\fX)$ denoting its $p$-th row and rank, respectively. 
The determinant of matrix $\fX$ is represented by $\text{det}(\fX)$. All logarithms are in base $2.$ The range of sequential positive integers from $1$ to $M$ is shown as $[\![M]\!].$ The set of whole numbers and non-negative real numbers are denoted by $\mathbb{N}_{0}$ and $\mathbb{R}_{+},$ respectively. The Gamma function is denoted by $\Gamma(x)$ for positive integer $x$ and is given by $\Gamma (x) = (x-1) ! \triangleq (x-1) \times \dots \times 1.$ We use the standard \emph{Bachmann-Landau asymptotic notation}, e.g., $o(\cdot)$ and $\mathcal{O}(\cdot).$ The $\ell_2$-norm and $\ell_{\infty}$-norm of $\fx$ are specified by $\norm{\mathbf{x}}$ and $\norm{\mathbf{x}}_{\infty},$ respectively. A hyper sphere of dimension $n,$ radius $r,$ and center $\mathbf{x}_0$ is defined as $\S_{\fx_0}(N,r) = \{\fx\in\mathbb{R}_{+}^N : \norm{\fx-\fx_0} \leq r \}.$ The $N$-dimensional cube whose center, edge length, and one of its corners are $(U/2,\ldots,U/2), U,$ and $\mathbf{0} = (0,\ldots,0),$ respectively, is symbolized by $\mathbb{Q}_{\f0}(N,U) = \{\fx \in \mathbb{R}_{+}^N : 0 \leq x_n \leq U, \forall \, n \in [\![N]\!] \}.$ The Kronecker delta as a function of two variables $\alpha$ and $\beta$ is specified by $\delta(\alpha,\beta)$ and is defined as $\delta_{\alpha,\beta} = 0$ if $\alpha \neq \beta$ and $\delta_{\alpha,\beta} = 1$ if $\alpha = \beta.$ 

\section{System Model and Coding Preliminaries}
In the following, we first present the affine Poisson channel in the context of olfactory-inspired MCs. Subsequently, we formally state the definition of identification capacity.

\subsection{Affine Poisson Channel}
\label{Sec.SysModel}
We assume that information about each task is represented by a mixture of $N$ different types of molecules. That is, the message is encoded within the signaling molecule type, and the transmitter controls and releases a mixture of molecules of different types with rate $x_n$ (molecules/channel use) over a time interval of $T_{\rm s},$ simultaneously, to encode the message, with $n \in [\![N]\!].$ The channel response for the $n$-th type of molecules is denoted by $v_n, \, n \in [\![N]\!]$ whereby the quantity $x_n v_n$ is the expected number of molecules that arrive at the receiver. When the number of molecules is large, the received signal follows Poisson statistics. Moreover, we assume that the receiver is equipped with $K$ different types of receptors, each may react to different types of molecules, i.e., a cross-reactive receptor array \cite{Jamali22}. Let $X_n \in\mathbb{R}_{+}$ and $Y _k \in \mathbb{N}_0$ denote RVs modeling the release rate of molecules of type $n$ by the transmitter and the number of bound receptors of type $k$ at the receiver, respectively. The channel input-output relation reads
\begin{align}
    Y_k = \text{Pois}\Big( \sum_{n=1}^N a_{k,n} v_n X_n + \lambda_k \Big) ,
\end{align}
where $a_{k,n} \in \mathbb{R}_{+}$ is known as the affinity between the receptor of type $k$ and molecule of type $n,$ which does not scale with the number of molecule type $N,$ and $\lambda_k$ stands for expected number of interfering molecules from undesired sources. The letter-wise channel model in compact matrix form is given by
\begin{align}
    \fY = \text{Pois}( \bar{\fA} \fX + \bm{\lambda} ) ,
\end{align}
where $\fX,\fY,$ and $\bm{\lambda}$ are input, output, and interference vector, i.e., $\fX_{N \times 1} = [X_1,\ldots,X_N]^{\dag}, \fY_{K \times 1} \allowbreak = [Y_1,\ldots,Y_K]^{\dag},$ and $\bm{\lambda}_{K \times 1} = [\lambda_1,\ldots,\lambda_K]^{\dag}.$ Furthermore, $\bar{\fA} \triangleq \fA \fV$ where $\fV$ and $\fA$ stand for the diagonal channel response matrix and affinity matrix, i.e., $\fV_{N \times N} = \text{diag}(\fv)$ with $\fv \triangleq (v_1,\ldots,v_N)$ and $\fA_{K \times N} = [a_{k,n}],$ for $k \in [\![K]\!], \, n \in [\![N]\!].$ We specify the rank of the affinity matrix $\fA_{K \times N}$ by $T,$ i.e., $T = \text{rank}(\fA) \leq \min\{K,N\} = K.$ Moreover, since the channel response matrix $\fV$ is a square diagonal matrix consisting of non-zero entries on the diagonal, it is full rank, i.e., $\text{rank}({\fV}) = N,$ hence it holds that $\text{rank}(\bar{\fA}) = \text{rank}(\fA) = T.$ The transition probability distribution for $K$ observations reads
\begin{equation}
\resizebox{.42\textwidth}{!}{
  $\begin{aligned}
    W^{K}(\fy|\fx) = \prod_{k=1}^{K} W(y_k|\fx)
    = \prod_{k=1}^{K} \frac{e^{- \bar{x}_k } \left( \bar{x}_k  + \lambda_k \right)^{y_k}}{{y_k}!} ,
    \vspace{-7mm}
  \end{aligned}$}
\end{equation}
where $\bar{x}_k \triangleq \sum_{n=1}^N a_{k,n} v_n x_n,$ $\fx = (x_1,\dots,x_N)$ and $\fy = (y_1,\dots,y_K)$ represent the transmitted codeword (i.e., a molecular mixture) and received receptor observations, respectively. The codewords are subjected to constraints $0 \leq x_n \leq C_{\rm max},\, \forall n \in [\![N]\!]$ and $N^{-1} \sum_{n=1}^{N} x_n \leq C_{\rm avg},$ respectively, where $C_{\rm max}, C_{\rm avg} > 0$ constrain the rate of molecule release per molecule type and across $N$ types of molecules in each mixture, respectively. To have a well-posed problem, we assume $\forall k \in [\![K]\!]$ and $\forall n \in [\![N]\!]$:
\begin{itemize}[leftmargin=*]
    \item $0 < v_{\text{min}} \leq v_n \leq v_{\text{max}},$ with $v_{\text{min}}$ and $v_{\text{max}}$ being constants.
    \item $0 < \lambda_{\text{min}} \leq \lambda_k \leq \lambda_{\text{max}},$ with $\lambda_{\text{min}}$ and $\lambda_{\text{max}}$ being constants.
    \item $0 \hspace{-.5mm} < \hspace{-.5mm} A_{\text{min}} \hspace{-.5mm} \leq \hspace{-.5mm} |a_{k,n}| \hspace{-.5mm} \leq \hspace{-.5mm} A_{\text{max}},$ holds for $k,n$ such that $a_{k,n} \neq 0$ where $A_{\text{min}}$ and $A_{\text{max}}$ being constants.
\end{itemize}

\subsection{Identification Coding}
The definition of an identification code for DAPC $\P_{\rm \fA}$ is formally presented in the following.
\begin{definition}
\label{Def.Affine_Identification_Code}
An $( N, M_N(R), \varepsilon_1, \varepsilon_2 )$-code for a DAPC $\P_{\rm \fA}$ with affinity matrix $\fA$ subject to average and peak constraints $C_{\rm avg} $ and $ C_{\rm max},$ respectively, and for integer $M_N(R)$ where $N$ and $R$ are the molecule types and coding rate, respectively, is defined as a system $(\mathbb{C}_0,\mathscr{Z})$ consisting a codebook $\mathbb{C}_0 = \{ \fc^i \}$ such that $0 \leq c_n^i \leq C_{\rm max}$ and $N^{-1} \sum_{n=1}^{N} c_n^i \leq C_{\rm avg},$ $\forall i \in [\![M]\!],\,\forall n \in [\![N]\!],$ and a collection of decoding regions $\mathscr{Z} = \{ \Z_i \}.$
Given a message $i \in [\![M]\!],$ the encoder transmits codeword $\fc^i,$ (i.e., a molecular mixture), the decoder's aim is to answer the below question: Was a desired message $j,$ (i.e., molecular mixture $j$), sent or not?
There are two errors that may occur: Rejection of the true message or acceptance of a false message, referred to as type I and type II errors, respectively, which read
\begin{align}
\label{Eq.Errors}
\hspace{-1mm}
    P_{e,1}(i) \hspace{-.5mm} = \hspace{-.5mm} 1 \hspace{-.3mm} - \hspace{-.6mm} \sum_{\Z_i} \hspace{-.5mm} W^K \big( \fy \, | \, \fc^i \big), P_{e,2}(i,j) \hspace{-.5mm} = \hspace{-.5mm} \sum_{\Z_j} \hspace{-.6mm} W^K \big( \fy \, | \, \fc^i \big). \hspace{-1mm}
\end{align}
It must holds $P_{e,1}(i) \leq \varepsilon_1$ and $P_{e,2}(i,j) \leq \varepsilon_2, \forall \, i,j \in [\![M]\!]$ such that $i \neq j, \allowbreak \, \forall \varepsilon_1, \allowbreak \varepsilon_2 \allowbreak > 0.$ A rate $R > 0$ is achievable if $\forall \varepsilon_1,\varepsilon_2 > 0$ and sufficiently large $N,$ there exists an identification code meeting the error conditions in \eqref{Eq.Errors}. Identification capacity for the DAPC $\P_{\rm \fA},$ is denoted by $\mathbb{C}_{\rm I}(\P_{\rm \fA}
),$ and is supremum of all achievable rates.
\qed
\end{definition}
\begin{remark}
The size of the codebook $M_N(R)$ is a sequence of monotonically non-decreasing function in $N.$ Moreover, we assume that the rank of the affinity matrix scales sub-linearly in the molecule types $N,$ i.e., $T = N^{\kappa\tau}$ with $\kappa,\tau \in (0,1].$ In particular, in this paper we assume the following function:
$$M_N(R) = 2^{(\kappa\tau N^{\kappa\tau}\log N)R} = 2^{(T\log T)R} = M_T(R).$$
Cf. Appendix~\ref{App.Codebook_Spectrum} for more discussions on the codebook sizes.
\end{remark}
\section{Identification Capacity of the DAPC}
\label{Sec.Res}
We start this section with presenting several key questions that we aim to answer, and are relevant in the design of a synthetic olfactory MC system. Subsequently, we provide our main theorem and draw various insights to answer the previously posed questions. Finally, we provide a sketch of the converse and achievability proofs of the capacity theorem. In particular, with the capacity analysis, we hope to shed light on the following questions:

\begin{itemize}[leftmargin=*]
    \item \textbf{Number of Identifiable Mixtures:} How many molecular mixtures (messages) can be reliably identified utilizing asymptotically large $N,$ i.e., different types of molecules?
    \item \textbf{Number of Required Receptor Types:} Does the number of receptors $K$ at the receiver need to scale linearly in $N$? Can we ensure reliable identification for $K = o(N)$?
    \item \textbf{Characteristics of Affinity Matrix:} What structure does the molecule-receptor affinity matrix must fulfill to guarantee a reliable identification? What are generic properties of an eligible affinity matrix, e.g., sparsity, rank, etc?
\end{itemize}

\subsection{Main Results}
\label{Subsec.Main_Results}
The identification capacity depends on the properties of the affinity matrix $\fA,$ and how they asymptotically scale with $N.$ Thus, before presenting our capacity results, we introduce a class of $\fA,$ which is defined by the below given four conditions. These conditions characterize the necessary requirements on $\fA,$ that allow reliable identification. Hence, these conditions not only provide a rigorous basis for the proposed capacity theorem, but also offer insights for choosing molecule and receptor types in practice.

\textbf{C1:} The number of rows of $\fA,$ i.e., the number of receptors $K,$ scales sub-linearly in the molecule types $N,$ i.e., $K = N^{\kappa},$ with $\kappa \in (0,1]$ being the receptors rate.

\textbf{C2:} Let $F_k$ be the number of non-zero entries in each row of $\fA$ and the following property: $F_k$ scales at most sub-linearly in $T,$ i.e., $F_k = \allowbreak \mathcal{O}(T^l) ,\, \forall k \in [\![K]\!],$ where $l \in [0,1)$ being the non-zero rate across all rows.

\textbf{C3:} We assume that the rank of the matrix $\fA$ scales sub-linearly in the number of observations $K,$ i.e., $T = K^{\tau} = N^{\kappa\tau}$ with $(\kappa + (1 - \delta_{\kappa,1})l)^{-1} \leq \tau \leq 1.$

\textbf{C4:} We assume that certain $T$ columns of the affinity matrix fulfill an orthogonality condition. Specifically, let us define:
\begin{itemize}[leftmargin=*]
    \item $\bar{\fA}_{\mathbb{G}_T}$ denotes the $K \times T$ sub-matrix of $\bar{\fA}$ consisting $\mathbb{G}_T$ columns with $\mathbb{G}_T \triangleq \allowbreak \{n_1,\ldots,n_T\}$ being the set of $T$ column indices with $n_1 < n_2 < \allowbreak \ldots \allowbreak < n_T.$
    \item $\bar{\bar{\fA}}_{\mathbb{G}_T} \triangleq \bar{\fA}_{\mathbb{G}_T}^{\dag} \bar{\fA}_{\mathbb{G}_T}$ and $f(\mathbb{G}_T) \triangleq \text{det}\big( \bar{\bar{\fA}}_{\mathbb{G}_T} \big) .$
    \item $\mathbb{G}_{T^*} \triangleq \underset{\mathbb{G}_T \in \{ \mathbb{G}_T' \subset [\![N]\!],\, |\mathbb{G}_T'| = T \} }{\argmax} \, f(\mathbb{G}_T) .$
    \item $\bar{\bar{\fA}}_{\mathbb{G}_{T^*}}^{n_t}$ indicates the $n_t$ column of matrix $\bar{\bar{\fA}}_{\mathbb{G}_{T^*}}.$
\end{itemize}
Then, we assume that the number of columns in the set $\mathbb{G}_{T^*}$ that are not orthogonal to each other is $ o(T),$ i.e., we have $\big| t'; \langle \bar{\bar{\fA}}_{\mathbb{G}_{T^*}}^{n_t'}.\bar{\bar{\fA}}_{\mathbb{G}_{T^*}}^{n_t} \rangle \neq 0 \big| = o(T), \forall n_t' \in \mathbb{G}_{T^*},\forall n_t \in \mathbb{G}_{T^*} \hspace{-.4mm} \setminus \hspace{-.3mm} n_t'.$

\vspace{2mm}
Now, we proceed to present our capacity theorem for the DAPC in the following theorem.

\begin{theorem}
\label{Th.Affine-Capacity}
Let $\P_{\rm \fA}$ be a DAPC with affinity matrix $\fA_{K \times N}$ fulfilling conditions \textbf{C1}-\textbf{C4} with $0 \leq l < 1/4.$ Then, the identification capacity of $\P_{\rm \fA}$ subject to average and peak molecule release rate constraints in Definition~\ref{Def.Affine_Identification_Code} in the super-exponential scale, i.e., $M_T(R) = 2^{(T\log T)R},$ reads
\begin{align}
    \label{Ineq.LU}
    \hspace{-2mm} \frac{1}{2} - \Big( \frac{\kappa}{4} + l \Big) \leq \mathbb{C}_{\rm I}(\P_{\rm \fA}) \leq \frac{1}{2} + \kappa + l + 2(1 - \delta_{\kappa,1})l . \hspace{-2mm}
\end{align}
\end{theorem}
\begin{proof}
The proof consists of achievability and converse and are provided in Subsections~\ref{Subsec.Achievability} and \ref{Subsec.Converse}, respectively. 
\end{proof}
\vspace{-1mm}
In the subsequent, we outline several key findings that can be derived from Theorem~\ref{Th.Affine-Capacity} and the accompanying proof.

\textbf{\textcolor{mycolor12}{Codebook Scale:}} Theorem~\ref{Th.Affine-Capacity} shows that the number of reliably identifiable messages (i.e., molecular mixtures) scales super-exponentially in the rank of the affinity matrix, i.e., $M_T(R) = 2^{(T \log T)R}.$ This is different to the Shannon capacity where $M_N(R) = 2^{NR}.$

\textbf{\textcolor{mycolor12}{Receptors Scale:}} 
Theorem~\ref{Th.Affine-Capacity} reveals that for reliable identification, the number of receptor types $K$ can scale sub-linearly in $N,$ i.e., $K_N(\kappa) = N^{\kappa}$ with $\kappa \in (3/4,1].$

\textbf{\textcolor{mycolor12}{Affinity Matrix:}} To ensure reliable identification, the affinity matrix $\fA$ must satisfy certain conditions. Through the proof of Theorem~\ref{Th.Affine-Capacity}, we have identified conditions \textbf{C1}-\textbf{C4}, which pose various conditions on the dimension (\textbf{C1}), sparsity (\textbf{C2}), rank (\textbf{C3}) and orthogonality (\textbf{C4}) of matrix $\fA$ to enable reliable identification of a super-exponentially large number of messages in $T$. 

\textbf{\textcolor{mycolor12}{Olfactory Capability:}}
Theorem~\ref{Th.Affine-Capacity} suggests that for $\kappa \in (3/4,1]$ still super-exponential number of messages are reliably identifiable conditioning that $T \to \infty$ as $N \to \infty.$ An operational meaning of this observation is that when the number of molecule type is sufficiently large, messages are identifiable even with relatively small $K.$

\begin{customcorollary}{1}[\textbf{Identification Capacity of Poisson Channel}]
\label{Cor.one}
The bounds in Theorem~\ref{Th.Affine-Capacity} recover the previously known results for the identification capacity of standard memoryless Poisson channel \cite[Th.~1]{Salariseddigh-TMBMC} by choosing $K=N,$ (i.e., $\kappa = 1$) and $\fA = \mathbf{I}$ (i.e., $l = 0$ and $T=N$). Then,
\begin{align}
    \frac{1}{4} \leq \mathbb{C}_{\rm I}(\P_{\rm \fA = \mathbf{I}}) \leq \frac{3}{2} .
\end{align}
\end{customcorollary}

\begin{proof}
The proofs of Corollary~\ref{Cor.one} follows by setting $l = 0$ and $\kappa = 1,$ respectively, into \eqref{Ineq.LU}.
\end{proof}
\vspace{-2mm}
\begin{customcorollary}{2}[\textbf{Identification Capacity of ISI Poisson Channel}]
\label{Cor.two}
The bounds in Theorem~\ref{Th.Affine-Capacity} recover the previously known results for the identification capacity of ISI Poisson channels \cite[Th.~1]{Salariseddigh_OJCOMS_23} where $F_k$ (which can be chosen to capture the number of ISI taps in Poisson channels with ISI) grows sub-linearly in $N,$ i.e., $F_k = N^l, \, \forall k \in [\![K]\!]$ with $l \in [0,1/4),$ and the number of receptors $K$ scales linearly in the molecule types $N,$ i.e., $\kappa = 1,$ and the affinity matrix $\fA$ has a Toeplitz structure with full rank, i.e., $T=N,$ then we recover the results in \cite{Salariseddigh_OJCOMS_23}. That is, the identification capacity reads
\begin{align*}
    \frac{1}{4} - l \leq \mathbb{C}_{\rm I}(\P_{\rm\fA}) \leq \frac{3}{2} + l .
\end{align*}
\end{customcorollary}
\begin{proof}
The proof follows by setting $\kappa = 1,$ into \eqref{Ineq.LU}.
\end{proof}
\vspace{-1mm}
In below, we provide technical contributions of the proof.
\vspace{-1mm}
\subsection{Achievability}
\label{Subsec.Achievability}
The achievability proof consists of the following two steps.
\begin{itemize}[leftmargin=*]
    \item \textbf{\textcolor{mycolor12}{Step 1:}} Codebook construction is demonstrated through the utilization of appropriate packing of hyper spheres.
    \item \textbf{\textcolor{mycolor12}{Step 2:}} We show this codebook guarantees an \textit{achievable} rate by proposing a decoder and proving that the corresponding type I and type II error probabilities vanish as $T \to \infty .$
\end{itemize}

\textbf{\textcolor{mycolor12}{Affine Codebook Construction:}}
We develop codewords that fulfill $0 \leq c_n^i \leq C_{\rm avg},$ $\forall \, i \in [\![M]\!] \,, \forall \, n \in [\![N]\!],$ which ensures both constraints in Definition~\ref{Def.Affine_Identification_Code}. In the following, instead of directly constructing the original codebook $\mathbb{C}_0 = \{ \fc^i \} \subset \mathbb{R}_{+}^N,$ we present a construction of a codebook called affine codebook which is a subspace in the receptor array. Let $\bar{\fc}_i$ be a codeword in the receptor array with elements $\bar{c}_k^i \triangleq \sum_{n=1}^N a_{k,n} v_n c_n^i ,\,\forall k \in [\![K]\!].$ Moreover, let $\bar{\fA} = \fU \mathbf{\Sigma} \fW^{\dag}$ denote SVD of matrix $\bar{\fA} \triangleq \fA\fV,$ where $\fU_{K \times K}$ and $\fW_{N \times N}$ are unitary matrices and $\mathbf{\Sigma}_{K \times N}$ is rectangular diagonal matrix consisting of $T$ non-zero eigenvalues of $\bar{\fA}$ on its first $T$ diagonal entries. With these notations, we formally define the original codebook, the affine codebook, and the reduced affine codebook as follows, respectively
\begin{itemize}[leftmargin=*]
    \item $\mathbb{C}_0 \triangleq \big\{ \fc_i \in \mathbb{R}_{+}^N:\; 0 \leq c_n^i \leq  C_{\rm avg} , \forall \, i \in [\![M]\!] , \forall \, n \in [\![N]\!] \big\} .$
    \item $\bar{\mathbb{C}}_0^{\fA} \triangleq \big\{ \bar{\fc}_{\fA}^i = \bar{\fA}\fc_i \in \mathbb{R}_{+}^K:\, \fc_i \in \mathbb{C}_0, \forall \, i \in [\![M]\!] \big\} .$
    \item $\bar{\bar{\mathbb{C}}}_0^{\fA} \triangleq \big\{ \bar{\bar{\fc}}_{\fA}^i = {\fU}_T^\dag\bar{\fc}_i \in \mathbb{R}_{+}^T:\, \bar{\fc}_i \in \bar{\mathbb{C}}_0^{\fA} , \forall \, i \in [\![M]\!]\big\} .$
\end{itemize}
where $\fU_T$ is a sub-matrix of $\fU$ containing its first $T$ columns.

Next, in order to study the impact of applying matrix $\fA$ on the input space, we establish a lemma provided in Appendix~\ref{App.Volume_under_Affine_Transformation}.

\textbf{\textcolor{mycolor12}{Rate Analysis:}} We use a packing arrangement of non-overlapping hyper spheres of radius $r_0 = \sqrt{T \epsilon_T}$ in the reduced affine codebook space $\bar{\bar{\mathbb{C}}}_0^{\fA}$ where $\epsilon_T = a / T^{( 2 - ( \kappa + 4l + b)) / 2},$ with $a > 0$ being a fixed constant and $b$ denoting an arbitrarily small constant. Let $\mathscr{S}$ denote a sphere packing, i.e., an arrangement of $M$ non-overlapping spheres $\S_{\bar{\bar{\fc}}_{\fA}^i}(T,r_0),\, i \in [\![M]\!],$ that are packed inside $\bar{\bar{\mathbb{C}}}_0^{\fA}.$ We require that the centers of disjoint spheres are inside $\bar{\bar{\mathbb{C}}}_0^{\fA}$ and have a non-empty intersection with it. The packing density for such arrangement is \cite{CHSN13} $\Updelta_T(\mathscr{S}) \triangleq \text{Vol}\big(\bigcup_{i=1}^{M}\S_{\bar{\bar{\fc}}_{\fA}^i}(T,r_0)\big) / \text{Vol}( \bar{\bar{\mathbb{C}}}_0^{\fA} ).$

A saturated packing argument is employed here, reminiscent of the one found in the Minkowski--Hlawka theorem for saturated packing discussed in \cite{CHSN13}. Specifically, consider a saturated packing of $\bigcup_{i=1}^{M_N(R)} \S_{\bar{\bar{\fc}}_{\fA}^i}(T,\sqrt{T\epsilon_T})$ spheres with radius $r_0 = \sqrt{T\epsilon_T}$ embedded within the reduced affine codebook $\bar{\bar{\mathbb{C}}}_0^{\fA}.$ Note that the density of such arrangement fulfills $2^{-T} \leq \Updelta_T(\mathscr{S}) \leq 2^{-0.599T},$ \cite[Sec.~IV]{Salariseddigh-TMBMC}. We assign a mixture to the center $\bar{\bar{\fc}}_{\fA}^i$ of each sphere. The reduced affine codewords satisfy $\bar{\bar{\fc}}_{\fA}^i \in \bar{\bar{\mathbb{C}}}_0^{\fA}.$ In general, volume of a hyper sphere with radius $r$ is $\text{Vol}\left(\S_{\bar{\fx}_{\fA}}(T,r)\right) = (\sqrt{\pi}r)^{T} \cdot \Gamma(T/2+1)$ \cite[Eq.~(16)]{CHSN13}.
Since centers of spheres lie inside $\bar{\bar{\mathbb{C}}}_0^{\fA},$ the total number of packed spheres, $M,$ is given by
\begin{align}
    \hspace{-2mm} M \hspace{-.6mm} = \hspace{-.6mm} \frac{\text{Vol}\big(\bigcup_{i=1}^{M}\S_{\bar{\bar{\fc}}_{\fA}^i}(T,r_0)\big)}{\text{Vol}(\S_{\bar{\bar{\fc}}_{\fA}^1}(T,r_0))} \hspace{-.6mm} \stackrel{(a)}{\geq} \hspace{-.6mm} 2^{-T} \hspace{-.5mm} \cdot \hspace{-.5mm} \frac{C_{\rm avg}^T \cdot \sqrt{\text{det}\big( \bar{\bar{\fA}}_{\mathbb{G}_{T^*}} \big)}}{\text{Vol}(\S_{\bar{\bar{\fc}}_{\fA}^1}(T,r_0))} , \hspace{-1.5mm}
    \label{Eq.M}
\end{align}
where $(a)$ uses the lower bound on the density and Lemma~\ref{Lem.Affine_Transformation}.

Exploiting the formula for volume of hyper sphere, Stirling's approximation, i.e., $\log T! = T \log T - T \log e + o(T)$ \cite[Sec.~IV]{Salariseddigh-TMBMC} substitution of $T$ with $\floor{T/2},$ $\Gamma(T/2 + 1) \geq \floor{T/2} !,$ and $r_0 = \allowbreak \sqrt{T \epsilon_T} = \allowbreak \sqrt{a} T^{(\kappa + 4l + b) / 4},$ we obtain
\begin{equation}
\hspace{-3mm}\resizebox{.45\textwidth}{!}{
  $\begin{aligned}
    \log M \hspace{-.5mm} \geq \hspace{-.5mm} \left( \frac{2 - (\kappa + 4l + b)}{4} \right) \hspace{-.5mm} T \log T \hspace{-.5mm} + \hspace{-.5mm} T \hspace{-.5mm} \left( \log \frac{C_{\text{avg}}}{\sqrt{ae}} \right) \hspace{-.5mm} + \hspace{-.5mm} \mathcal{O}(T)
    , \hspace{-1mm}
  \end{aligned}$}
  \label{Eq.Log_M_Compact}
\end{equation}
see Appendix~\ref{App.Rate_Analysis} for details. Observe that the dominant term in \eqref{Eq.Log_M_Compact} is of order $T \log T.$ Hence, to have a non-zero rate $R,$ the codebook size $M$ must scales as $2^{(T \log T)R}.$ Therefore,
\begin{equation*}
\resizebox{.49\textwidth}{!}{
  $\begin{aligned}
    R \geq \frac{1}{T \log T} \left[ \left( \frac{2 - (\kappa + 4l + b)}{4} \right) T \log T + T \log \left( \frac{C_{\text{avg}}}{\sqrt{ae}} \right) + \mathcal{O}(T) \right] ,
  \end{aligned}$}
\end{equation*}
which tends to $(2 - (\kappa + 4l))/4$ when $T \to \infty$ and $b \rightarrow 0.$

\textbf{\textcolor{mycolor12}{Encoding:}}
Given message $i\in [\![M]\!],$ transmit $\fx = \fc^i.$

\textbf{\textcolor{mycolor12}{Decoding:}}
Let $\varepsilon_1, \varepsilon_2, \zeta_0, \zeta_1 > 0$ be arbitrarily small constants. Next, for conducting concise analysis, we use definitions in Appendix~\ref{App.Decoding_Conventions}. Now, let us define the \emph{decoding threshold} as follows $\psi_T = 4a / 3T^{( 2 - ( \kappa + 4l + b)) / 2},$ where $a$ is a fixed constant, $b$ is an arbitrarily small constant. To identify if message $j \in \M$ was sent, the decoder checks whether $\mathbf{y}$ belongs to the decoding set $\Z_j = \big\{ \fy \in \mathbb{N}_0^K \;:\, | Z(\fy;\fc^j) | \leq \psi_T \big\},$ with
\begin{align}
    \label{Eq.DM}
    Z(\fy;\fc^j) = T^{-1} \sum_{k \in \mathbb{E}_T} \big( y_k - \big( \bar{c}_k^j + \lambda_k \big) \big)^2 - y_k ,
\end{align}
where $\mathbb{E}_T \triangleq \allowbreak \{k_1,\ldots,k_T\}$ is a set of $T$ row indices whose corresponding rows are linearly independent.

The detailed proof for vanishing error probabilities is provided in Appendix~\ref{App.Extensive_Error_Analysis}. In the following, we provide a sketch of the proof and adopted key techniques.

\textbf{\textcolor{mycolor12}{Calculation of Type I Error:}}
The type I errors occur when the transmitter sends $\fc^i,$ yet $\fY\notin\Z_i.$ For every $i \in [\![M]\!],$ the type I error probability is bounded by 
\begin{align*}
    P_{e,1}(i) = \Pr\big( \fY(i) \in \Z_i^c \big) = \Pr\big( \big| Z(\fY(i),\fc^i) \big| > \psi_T \big) . 
\end{align*}
To bound $P_{e,1}(i),$ we apply Chebyshev's inequality, namely
\begin{equation*}
\vspace{-1mm} \resizebox{.49\textwidth}{!}{
  $\begin{aligned}
    \Pr\big(\big| Z(\fY(i); \fc^i) \hspace{-.6mm} - \hspace{-.6mm}\mathbb{E} \big[ Z(\fY(i); \fc^i) \big] \big| \hspace{-.6mm} > \hspace{-.6mm} \psi_T \big) \hspace{-.6mm} \leq \hspace{-.6mm} \frac{\text{Var} \big[ Z(\fY(i);\fc^i) \big]}{\psi_T^2} .
  \end{aligned}$}
\end{equation*}
Note $\mathbb{E}\big[ Z(\fY(i); \fc^i) \big] = 0.$ Next, exploiting the property of the statistical independence of observations, $\text{Var} \big[ Y_k(i) \big] = \bar{c}_k^i + \lambda_k \geq 0,$ and Lemma~\ref{Lem.MGF}, we bound the type I error as follows
\begin{align}
    P_{e,1}(i) \leq \frac{7U_{k,C_{\text{avg}}}}{T \psi_T^2} \triangleq \zeta_1 = \frac{\mathcal{O}(T^{4l})}{T \psi_T^2} = \frac{\mathcal{O}(1)}{T^{b}} \leq \varepsilon_1 ,
\end{align}
for sufficiently large $T$ and arbitrarily small $\varepsilon_1 > 0.$

\textbf{\textcolor{mycolor12}{Calculation of Type II Error:}}
We examine type II errors, i.e., when $\fY\in\Z_j$ while the transmitter sent $\fc^i$ with $i \neq j.$ Then, for every $i,j \in [\![M]\!],$ the type II error probability reads
\begin{align}
    P_{e,2}(i,j) = \Pr \big( \big| Z(\fY(i);\fc^j) \big| \leq \psi_T \big) ,
    \label{Eq.Pe2G}
\end{align}
where $Z(\fY(i);\fc^j) = \omega - \phi ,$ with $$\phi \triangleq T^{-1} \sum_{k \in \mathbb{E}_T} Y_k(i) \quad \text{ and } \quad \omega \triangleq T^{-1} \sum_{k \in \mathbb{E}_T} ( \omega_i + \omega_{i,j} )^2$$
with $\omega_i \hspace{-.4mm} \triangleq \hspace{-.4mm} Y_k(i) - ( \hspace{-.5mm} \bar{c}_k^i + \lambda_k )$ and $\omega_{i,j} \triangleq \bar{c}_k^i - \bar{c}_k^j.$ Observe that $\omega = \omega_1 + \omega_2$ where $\omega_1 \triangleq T^{-1} \sum_{k \in \mathbb{E}_T} \omega_i^2 + \omega_{i,j}^2$ and $\omega_2 \allowbreak \triangleq \allowbreak 2T^{-1} \allowbreak \sum_{k \in \mathbb{E}_T} \allowbreak \omega_i \omega_{i,j}.$ Next, let $\Omega_0 = \{ | \omega_2 | > \psi_T \}$ and $\Omega_1 = \{ \omega_1 - \phi \leq 2\psi_T \}.$ Now, expanding $\omega$ and exploiting the reverse triangle inequality, we obtain $P_{e,2}(i,j) \leq \Pr( \omega - \phi \leq \psi_T ) .$ Hence, utilizing the principle of total probability on event $\D = \big\{ \omega - \phi \leq \psi_T \big\}$ over $(\Omega_0,\Omega_0^c)$ we obtain $$P_{e,2}(i,j) \leq \Pr(\Omega_0 ) + \Pr ( \D \cap \Omega_0^c ) \leq \Pr(\Omega_0) + \Pr(\Omega_1 ).$$
Now, we apply the Chebyshev's inequality to bound $\Pr ( \Omega_0 )$
\begin{align}
    \label{Ineq.E_0_2}
    \Pr(\Omega_0) \leq \frac{16 C_{\text{avg}}^2 (F_k \cdot A_{\text{max}} v_{\text{max}})^2 A_k }{T \psi_T^2} \triangleq \zeta_0 \stackrel{(a)} = \frac{\mathcal{O}(1)}{T^{l + b}} ,
\end{align}
for sufficiently large $T,$ where $\zeta_0 > 0$ is an arbitrarily small constant and $(a)$ exploits condition \textbf{C2} in Subsection~\ref{Subsec.Main_Results}, i.e., $F_k = \mathcal{O}(T^l), \, \forall k \in [\![K]\!].$ Cf. Appendix~\ref{App.Extensive_Error_Analysis} for details. Next, exploiting Lemma~\ref{Lem.Minimum_Distance_Affine} and standard techniques, we get
\begin{align}
    \label{Ineq.Omega_1}
    \Pr(\Omega_1) \leq \frac{7U_{k,C_{\rm avg}}}{T\psi_T^2} \triangleq \zeta_1 = \frac{\mathcal{O}(T^{4l})}{T \psi_T^2}
    = \frac{\mathcal{O}(1)}{T^{b}} ,
 \end{align}
for sufficiently large $T,$ where $\zeta_1 > 0$ is arbitrarily small constant. See Appendix~\ref{App.Extensive_Error_Analysis} for details. Thereby, the type II error probability is upper bounded by $$P_{e,2}(i,j) \leq \allowbreak \Pr(\Omega_0) \allowbreak + \allowbreak \Pr(\Omega_1) \allowbreak \leq \allowbreak \zeta_0 \allowbreak + \allowbreak \zeta_1 \allowbreak \leq \allowbreak \varepsilon_2.$$

We have thus shown that $\forall \varepsilon_1,\varepsilon_2>0$ and sufficiently large $N,$ there is an $(N, M_N(R), \varepsilon_1, \varepsilon_2)$ affine Poisson identification code. This completes the achievability proof of Theorem~\ref{Th.Affine-Capacity}.

\subsection{Upper Bound (Converse Proof)}
\label{Subsec.Converse}
The converse proof consists of the following two main steps.
\begin{itemize}[leftmargin=*]
    \item \textbf{\textcolor{mycolor12}{Step~1:}} Lemma~\ref{Lem.Converse} guarantees that for any achievable rate, the distance between shifted symbols for two different molecular mixture is lower bounded by a threshold.
    \item\textbf{\textcolor{mycolor12}{Step~2:}} Employing Lemma~\ref{Lem.Converse}, we derive an upper bound on the codebook size of achievable identification codes.
\end{itemize}
Before proceeding, for the purpose of conducting a brief analysis, we set the following conventions:
\begin{itemize}[leftmargin=*]
    \item $\bar{\fY}(i) \triangleq K^{-1} \sum_{k=1}^K Y_k(i) , \, \forall i \in [\![M]\!].$ 
    \item $\bar{A}_{\rm max} \triangleq \underset{k \in \mathbb{E}_T}{\max}\, [ F_k \cdot \left( A_{\text{max}} v_{\text{max}} C_{\rm max} \hspace{-.4mm} + \hspace{-.4mm} \lambda_{\text{max}} \right) ] = \mathcal{O}(T^l) .$
    \item $\mathbb{C}_{0,\text{\tiny conv}} \hspace{-.6mm} \triangleq \hspace{-.6mm} \big\{ \fc_i \in \mathbb{R}_{+}^N: \hspace{-.5mm} 0 \hspace{-.5mm} \leq \hspace{-.5mm} c_n^i \hspace{-.5mm} \leq \hspace{-.5mm} C_{\text{max}} , \forall \, i \in [\![M]\!] , \forall \, n \in [\![N]\!] \big\}.$
    \item $\bar{\mathbb{C}}_{0,\text{\tiny conv}}^{\fA} \triangleq \big\{ \bar{\fc}_{\fA}^i = \bar{\fA}\fc_i \in \mathbb{R}_{+}^K:\, \fc_i \in \mathbb{C}_{0,\text{\tiny conv}} , \forall \, i \in [\![M]\!] \big\}.$
    \item $\bar{\bar{\mathbb{C}}}_{0,\text{\tiny conv}}^{\fA} \hspace{-.5mm} \triangleq \big\{ \bar{\bar{\fc}}_{\fA}^i = {\fU}_T^\dag\bar{\fc}_i \in \mathbb{R}_{+}^T:\, \bar{\fc}_i \in \bar{\mathbb{C}}_{0,\text{\tiny conv}}^{\fA} , \forall \, i \in [\![M]\!]\big\}.$
    \item $d_k^{i} = \bar{c}_k^{i} + \lambda_k \,, \forall k \in [\![K]\!].$
\end{itemize}
\begin{customlemma}{7}[\textbf{Shifted Symbol's Distance}]
\label{Lem.Converse}
Suppose that $R$ is an achievable identification rate for the DAPC $\P_{\rm \fA}.$ Consider a sequence of $( N, M_N(R), \varepsilon_1^{(N)}, \varepsilon_2^{(N)})$ codes $(\mathbb{C}_{0,\text{\tiny conv}}^{(N)},\mathscr{Z}^{(N)})$ where $K = N^{\kappa}$ such that $\varepsilon_1^{(N)}$ and $\varepsilon_2^{(N)}$ tend to zero as $N \rightarrow \infty.$ Then, given a sufficiently large $N,$ the associated affine codebook $\bar{\mathbb{C}}_{0,\text{\tiny conv}}^{\fA}$ satisfies the following property. For every pair of affine codewords $\bar{\fc}_{\fA}^{i_1}$ and $\bar{\fc}_{\fA}^{i_2},$ inside $\bar{\mathbb{C}}_{0,\text{\tiny conv}}^{\fA},$ there exists at least one $k' \in [\![K]\!]$ such that $|1 - d_{k'}^{i_2}/d_{k'}^{i_1}| > \theta_T, \, \forall i_1,i_2 \in [\![M]\!],$ 
with $i_1 \neq i_2$ and $\theta_T \triangleq C_{\rm max}/T^{\kappa(1-l\delta_{\kappa,1}) + 2l + b},$ where $b>0$ is an arbitrarily small constant.
\end{customlemma}
\begin{proof}
    The proof of Lemma~\ref{Lem.Converse} is provided in Appendix~\ref{App.Converse_Lemma}.
\end{proof}
Next, we use Lemma~\ref{Lem.Converse} to prove the upper bound on the identification capacity. Observe that Lemma~\ref{Lem.Converse} implies that there exists at least one letter $k' \in [\![K]\!]$ such that
\vspace{-1mm}
\begin{align*}
    \| \bar{\bar{\fc}}_{\fA}^{i_1} - \bar{\bar{\fc}}_{\fA}^{i_2} \| = \| \bar{\fc}_{\fA}^{i_1} - \bar{\fc}_{\fA}^{i_2} \| \geq \big| d_{k'}^{i_1} - d_{k'}^{i_2} \big| \stackrel{(a)}{>} \theta_T d_{k'}^{i_1} \stackrel{(b)}{>} \lambda_{\text{min}} \theta_T ,
\end{align*}
where $(a)$ uses Lemma~\ref{Lem.Converse} and $(b)$ exploits $d_k^{i_z} = \bar{c}_k^{i_z} + \lambda_k \geq \lambda_{\text{min}}, \, \forall k \in [\![K]\!]$ and for $z \in \{1,2\}$ with $\bar{c}_k^{i_z} \geq 0.$ Thus, we can define an arrangement of non-overlapping spheres $\S_{\bar{\bar{\fc}}_{\fA}^i}(T,\lambda_{\text{min}} \theta_T)$ inside $\bar{\bar{\mathbb{C}}}_0^{\fA}.$ Next, we provide the derivation of the upper bounds on achievable identification rate, $R,$ for $0 < \kappa \leq 1$ in Appendix~\ref{App.Derivation_Upper_Bounds}. Thereby, we obtain
\begin{align*}
    R \leq \frac{1}{2} + \kappa + l + 2(1 - \delta_{\kappa,1})l
    = \begin{cases}
    \frac{1}{2} + \kappa + 3l & 0 < \kappa < 1 , \\ \frac{3}{2} + l & \kappa = 1 .
    \end{cases}
\end{align*}
This completes the converse proof of Theorem~\ref{Th.Affine-Capacity}. 

\section{Conclusion}
\label{Sec.Conclusion}
We investigated the identification problem for affine DAPC, which can be used to model molecular mixture olfactory communications and include as special cases the conventional ISI-free and ISI channels. We proved that reliable identification can be accomplished with a super-exponentially large mixture codebook, $M = 2^{(T \log T)R},$ even when the number of receptor types scales sub-linearly with $N.$ In addition, we derived lower and upper bounds on achievable identification rate, $R,$ as functions of the parameters of the affinity matrix. Our results have the potential to be expanded in various directions such as multi user extension, random affinity matrix, finite molecule types analysis and practical code construction.

\section*{Acknowledgements}
This work is supported in part by the Federal Ministry of Research, Technology and Space of Germany (BMFTR) within the project IoBNT no. 16KIS1992 and no. 16KIS1988, LOEWE initiative (Hesse, Germany within emergenCITY center (LOEWE/1/12/519/03/05.001(0016)/72), and German Research Foundation (DFG) – GRK 2950 – Project-ID 509922606. H. Boche acknowledges the financial support by the BMFTR program of “Souver\"{a}n. Digital. Vernetzt.”. Joint project 6G-life, project identification number: 16KISK002 and BMFTR project Post Shannon Communication (NEWCOM), no. 16KIS1003K.

\vspace{-4mm}

\section*{}
\bibliographystyle{IEEEtran}
\bibliography{Lit-Cap-New2}

\begin{thebibliography}{10}
\providecommand{\url}[1]{#1}
\csname url@samestyle\endcsname
\providecommand{\newblock}{\relax}
\providecommand{\bibinfo}[2]{#2}
\providecommand{\BIBentrySTDinterwordspacing}{\spaceskip=0pt\relax}
\providecommand{\BIBentryALTinterwordstretchfactor}{4}
\providecommand{\BIBentryALTinterwordspacing}{\spaceskip=\fontdimen2\font plus
\BIBentryALTinterwordstretchfactor\fontdimen3\font minus \fontdimen4\font\relax}
\providecommand{\BIBforeignlanguage}[2]{{%
\expandafter\ifx\csname l@#1\endcsname\relax
\typeout{** WARNING: IEEEtran.bst: No hyphenation pattern has been}%
\typeout{** loaded for the language `#1'. Using the pattern for}%
\typeout{** the default language instead.}%
\else
\language=\csname l@#1\endcsname
\fi
#2}}
\providecommand{\BIBdecl}{\relax}
\BIBdecl

\bibitem{J85}
J.~J\'aJ\'a, ``Identification is {E}asier {T}han {D}ecoding,'' in \emph{Ann. Symp. Found. Comp. Scien.}, 1985, pp. 43--50.

\bibitem{Salariseddigh_IT}
M.~J. Salariseddigh, U.~Pereg, H.~Boche, and C.~Deppe, ``{D}eterministic {I}dentification {O}ver {C}hannels {W}ith {P}ower {C}onstraints,'' \emph{IEEE Trans. Inf. Theory}, vol.~68, no.~1, pp. 1--24, 2022.

\bibitem{AD89}
R.~{Ahlswede} and G.~{Dueck}, ``Identification via {C}hannels,'' \emph{IEEE Trans. Inf. Theory}, vol.~35, no.~1, pp. 15--29, 1989.

\bibitem{S48}
C.~E. {Shannon}, ``A {M}athematical {T}heory of {C}ommunication,'' \emph{Bell Sys. Tech. J.}, vol.~27, no.~3, pp. 379--423, 1948.

\bibitem{Colomer_2025}
\BIBentryALTinterwordspacing
P.~Colomer, C.~Deppe, H.~Boche, and A.~Winter, ``Deterministic {I}dentification {O}ver {C}hannels with {F}inite {O}utput: {A} {D}imensional {P}erspective on {S}uperlinear {R}ates,'' \emph{IEEE Trans. Inf. Theory}, p. 1–1, 2025. [Online]. Available: \url{http://dx.doi.org/10.1109/TIT.2025.3531301}
\BIBentrySTDinterwordspacing

\bibitem{Salariseddigh23_BSC_Future_Internet}
\BIBentryALTinterwordspacing
M.~J. Salariseddigh, O.~Dabbabi, C.~Deppe, and H.~Boche, ``Deterministic {K}-{I}dentification for {F}uture {C}ommunication {N}etworks: {T}he {B}inary {S}ymmetric {C}hannel {R}esults,'' \emph{Future {I}nternet}, vol.~16, no.~3, 2024. [Online]. Available: \url{https://www.mdpi.com/1999-5903/16/3/78}
\BIBentrySTDinterwordspacing

\bibitem{6G_PST}
J.~A. Cabrera, H.~Boche, C.~Deppe, R.~F. Schaefer, C.~Scheunert, and F.~H. Fitzek, ``6{G} and {T}he {P}ost-{S}hannon {T}heory,'' in \emph{Shaping Future 6G Networks: Needs, Impacts and Technologies}, N.~O. Frederiksen and H.~Gulliksen, Eds.\hskip 1em plus 0.5em minus 0.4em\relax Hoboken, NJ, United States: Wiley-Blackwell, 2021.

\bibitem{6G+}
W.~Haselmayr, A.~Springer, G.~Fischer, C.~Alexiou, H.~Boche, P.~A. Hoeher, F.~Dressler, and R.~Schober, ``Integration of {M}olecular {C}ommunications {I}nto {F}uture {G}eneration {W}ireless {N}etworks,'' in \emph{Proc. 6G Wireless Summit., Finland}, 2019.

\bibitem{Varshney16}
K.~R. Varshney and L.~R. Varshney, ``Olfactory {S}ignal {P}rocessing,'' \emph{Digital Signal Processing}, vol.~48, pp. 84--92, 2016.

\bibitem{Buck05}
L.~B. Buck, ``Unraveling {T}he {S}ense of {S}mell ({N}obel {L}ect.),'' \emph{Angew. Chem. Int. Ed.}, vol.~44, no.~38, pp. 6128--6140, 2005.

\bibitem{Bushdid14}
C.~Bushdid, M.~O. Magnasco, L.~B. Vosshall, and A.~Keller, ``Humans {C}an {D}iscriminate {M}ore {T}han 1 {T}rillion {O}lfactory {S}timuli,'' \emph{Science}, vol. 343, no. 6177, pp. 1370--1372, 2014.

\bibitem{Meister14}
\BIBentryALTinterwordspacing
M.~Meister, ``Can {H}umans {R}eally {D}iscriminate 1 {T}rillion {O}dors?'' 2014. [Online]. Available: \url{https://arxiv.org/abs/1411.0165}
\BIBentrySTDinterwordspacing

\bibitem{Meister15}
\BIBentryALTinterwordspacing
------, ``On {T}he {D}imensionality of {O}dor {S}pace,'' \emph{eLife}, vol.~4, p. e07865, jul 2015. [Online]. Available: \url{https://doi.org/10.7554/eLife.07865}
\BIBentrySTDinterwordspacing

\bibitem{Jamali22}
V.~Jamali, H.~M. Loos, and et~al., ``{O}lfaction-{I}nspired {MC}s: {M}olecule {M}ixture {S}hift {K}eying and {C}ross-{R}eactive {R}eceptor {A}rrays,'' \emph{IEEE Trans. Commun.}, vol.~71, no.~4, pp. 1894--1911, 2023.

\bibitem{Salariseddigh-TMBMC}
M.~J. Salariseddigh, V.~Jamali, U.~Pereg, H.~Boche, C.~Deppe, and R.~Schober, ``Deterministic {I}dentification {F}or {M}olecular {C}ommunications {O}ver {T}he {P}oisson {C}hannel,'' \emph{IEEE Trans. Mol. Biol. Multi-Scale Commun.}, vol.~9, no.~4, pp. 408--424, 2023.

\bibitem{Salariseddigh_OJCOMS_23}
------, ``Deterministic {K}-{I}dentification {F}or {MC} {P}oisson {C}hannel {W}ith {I}nter-{S}ymbol {I}nterference,'' \emph{IEEE Open J. Commun. Soc.}, pp. 1--1, 2024.

\bibitem{Jain_86}
R.~Jain and S.~Routhier, ``Packet {T}rains--{M}easurements and a {N}ew {M}odel for {C}omputer {N}etwork {T}raffic,'' \emph{IEEE J. Sel. Areas Commun.}, vol.~4, no.~6, pp. 986--995, 1986.

\bibitem{CHSN13}
J.~H. Conway and N.~J.~A. Sloane, \emph{Sphere {P}ackings, {L}attices and {G}roups}.\hskip 1em plus 0.5em minus 0.4em\relax New York, NY, USA: Springer, 2013.

\bibitem{Gover10}
E.~Gover and N.~Krikorian, ``Determinants and {T}he {V}olumes of {P}arallelotopes and {Z}onotopes,'' \emph{Linear Algebra and its Applications}, vol. 433, no.~1, pp. 28--40, 2010.

\bibitem{F66}
W.~Feller, \emph{An {I}ntroduction to {P}robability {T}heory and {I}ts {A}pplications}.\hskip 1em plus 0.5em minus 0.4em\relax John Wiley \& Sons, 1966.

\bibitem{Beals10}
R.~Beals and R.~Wong, \emph{Special {F}unctions: {A} {G}raduate {T}ext}.\hskip 1em plus 0.5em minus 0.4em\relax Cambridge University Press, 2010, vol. 126.

\bibitem{Garling07}
D.~J. Garling, \emph{Inequalities: {A} {J}ourney into {L}inear {A}nalysis}.\hskip 1em plus 0.5em minus 0.4em\relax Cambridge University Press, 2007.

\bibitem{Fiedler09}
M.~Fiedler, ``Suborthogonality and {O}rthocentricity of {M}atrices,'' \emph{Linear Algebra and its Applications}, vol. 430, no.~1, pp. 296--307, 2009.

\bibitem{Mitrinovic13}
D.~Mitrinovic, J.~Pecaric, and A.~Fink, \emph{Classical and {N}ew {I}nequalities in {A}nalysis}.\hskip 1em plus 0.5em minus 0.4em\relax Dordrecht, The Netherlands: Springer, 2013, vol.~61.

\bibitem{Sarkar24}
\BIBentryALTinterwordspacing
A.~Sarkar and B.~K. Dey, ``{I}dentification {O}ver {P}ermutation {C}hannels,'' \emph{arXiv:2405.09309}, 2024. [Online]. Available: \url{http://arxiv.org/abs/2405.09309.pdf}
\BIBentrySTDinterwordspacing

\bibitem{Salariseddigh_ICC}
M.~J. Salariseddigh, U.~Pereg, H.~Boche, and C.~Deppe, ``{D}eterministic {I}dentification {O}ver {C}hannels {W}ith {P}ower {C}onstraints,'' in \emph{IEEE Int. Conf. Commun.}, 2021, pp. 1--6.

\bibitem{Salariseddigh23_BSC_GC23}
O.~Dabbabi, M.~J. Salariseddigh, C.~Deppe, and H.~Boche, ``{D}eterministic {K}-{I}dentification {F}or {B}inary {S}ymmetric {C}hannel,'' in \emph{IEEE Glob. Commun. Conf.}, 2023, pp. 4381--4386.

\bibitem{Salariseddigh_ITW}
M.~J. Salariseddigh, U.~Pereg, H.~Boche, and C.~Deppe, ``Deterministic {I}dentification {O}ver {F}ading {C}hannels,'' in \emph{IEEE Inf. Theory Wksp.}, 2021, pp. 1--5.

\bibitem{Salariseddigh_22_ITW}
M.~Spahovic, M.~J. Salariseddigh, and C.~Deppe, ``Deterministic {K}-{I}dentification {F}or {S}low {F}ading {C}hannels,'' in \emph{IEEE Inf. Theory Wksp.}, 2023, pp. 353--358.

\bibitem{Salariseddigh_Binomial_ISIT}
M.~J. Salariseddigh, V.~Jamali, H.~Boche, C.~Deppe, and R.~Schober, ``Deterministic {I}dentification {F}or {MC} {B}inomial {C}hannel,'' in \emph{IEEE Int. Symp. Inf. Theory}, 2023, pp. 448--453.

\bibitem{Salariseddigh-ICC23}
M.~J. Salariseddigh, V.~Jamali, U.~Pereg, H.~Boche, C.~Deppe, and R.~Schober, ``{D}eterministic {I}dentification {F}or {MC} {ISI}-{P}oisson {C}hannel,'' in \emph{IEEE Intl. Conf. Commun.}, 2023, pp. 6108--6113.

\bibitem{Wiese22}
M.~Wiese, W.~Labidi, C.~Deppe, and H.~Boche, ``Identification {O}ver {A}dditive {N}oise {C}hannels in {T}he {P}resence of {F}eedback,'' \emph{IEEE Trans. Inf. Theory}, vol.~69, no.~11, pp. 6811--6821, 2023.

\end{thebibliography}

\clearpage

\appendices

\renewcommand{\thesectiondis}[2]{\Alph{section}:}
\section{Volume under Affine Transformation}
\label{App.Volume_under_Affine_Transformation}

\begin{customlemma}{1}
\label{Lem.Affine_Transformation}
Let $\bar{\fA} \triangleq \fA \fV$ where $\fA$ and $\fV$ are the affinity matrix with rank $T \leq K,$ and the channel response matrix. Moreover, let $\bar{\mathbb{C}}_0^{\fA}$ be the affine codebook generated under the matrix $\bar{\fA}$ in the space $\mathbb{R}_{+}^K$ and the assumption provided in condition $\textbf{C4}$ in Subsection \ref{Subsec.Main_Results} hold. Then,
\begin{align}
    \label{Eq.Vol_Zonotope}
    \text{Vol}_T(\bar{\mathbb{C}}_0^{\fA}) = \text{Vol}(\bar{\bar{\mathbb{C}}}_0^{\fA}) & = C_{\rm avg}^T \cdot \sum_{\substack{\mathbb{G}_T \subset [\![N]\!]:\,\\ |\mathbb{G}_T| = T}} \sqrt{\text{det}\big( \bar{\bar{\fA}}_{\mathbb{G}_T} \big)}
    \nonumber\\&
    \geq C_{\rm avg}^T \cdot \sqrt{\text{det}\big( \bar{\bar{\fA}}_{\mathbb{G}_{T^*}} \big)} .
  \end{align}
\end{customlemma}
\begin{proof}
The proof directly follows from \cite[Propos.~3.3]{Gover10} and is based on establishing a proper decomposition for an affine codebook $\bar{\bar{\mathbb{C}}}_0^{\fA}$ of rank $T$ in $\mathbb{R}^T.$
\end{proof}

\renewcommand{\thesectiondis}[2]{\Alph{section}:}
\section{Rate Analysis}
\label{App.Rate_Analysis}
In the following, we provide detailed derivations for establishment of the lower bound on the logarithm of the total number of packed spheres, $\log M.$ Observe that in order to obtain the simplified lower bound given in \eqref{Eq.Log_M_Compact}, we recall \eqref{Eq.M} and take logarithm from both side to obtain the following
\begin{align*}
    & \log M
    \nonumber\\&
    \geq \log \Big( C_{\rm avg}^T \cdot \sqrt{\text{det}\big( \bar{\bar{\fA}}_{\mathbb{G}_{T^*}} \big)} \Big) - \log \big(\text{Vol}\left(\S_{\bar{\bar{\fc}}^1}(T,r_0)\right) \big) - T
    \nonumber\\&
    \stackrel{(a)}{\geq} \log \sqrt{\text{det}\big( \bar{\bar{\fA}}_{\mathbb{G}_{T^*}} \big)} \hspace{-.5mm} + \hspace{-.5mm} T \log C_{\rm avg} + T \log( 1/\sqrt{\pi} r_0) 
    \nonumber\\&
    + \log \Gamma (( T/2) + 1 ) - T
    \stackrel{(b)}{\geq} \log \sqrt{\text{det}\big( \bar{\bar{\fA}}_{\mathbb{G}_{T^*}} \big)} \hspace{-.5mm} + \hspace{-.5mm} T \log C_{\rm avg} \hspace{-.5mm} 
    \nonumber\\&
    - \hspace{-.5mm} T \log r_0 \hspace{-.5mm} + \hspace{-.5mm} \floor{T/2} \log \floor{T/2} \hspace{-.5mm} - \hspace{-.5mm} \floor{T/2} \log e \hspace{-.5mm} + \hspace{-.5mm} o \big( \floor{T/2} \big) \hspace{-.5mm} - \hspace{-.5mm} T ,
\end{align*}
where $(a)$ exploits the formula for volume of a hyper sphere and $(b)$ follows by exploiting Stirling's approximation, i.e., $\log T! = T \log T - T \log e + o(T)$ \cite[P.~52]{F66} with substitution of $T$ with $\floor{T/2} \in \mathbb{Z},$ and since
\begin{align}
    \label{Ineq.Gamma_LB}
    \Gamma (( T/2) + 1 ) \stackrel{(a)}{\geq} \floor{T/2} \cdot \Gamma \big( \floor{T/2} \big) 
    \stackrel{(b)}{\triangleq} \floor{T/2} ! ,
\end{align}
where $(a)$ holds by $\Gamma (( T/2) + 1 ) = (T/2) \cdot \Gamma (T/2)$ for real $T/2$ \cite{Beals10}, $\floor{T/2} \leq T/2$ and the monotonicity of the Gamma function \cite{Beals10} which holds for $T \geq 4.$ $(b)$ holds by the definition of the Gamma function. Now, for $r_0 = \allowbreak \sqrt{T \epsilon_T} = \allowbreak \sqrt{a} T^{(\kappa + 4l + b) / 4}, $ we obtain
\begin{align}
    & \log M \stackrel{(a)}{\geq} \log \sqrt{\text{det}\big( \bar{\bar{\fA}}_{\mathbb{G}_{T^*}} \big)} - \Big( \frac{\kappa + 4l + b}{4} \Big) T \log T 
    \nonumber\\&
    + T \log \Big( \frac{C_{\rm avg}}{\sqrt{ae}} \Big) + ( (T/2) - 1 ) \log ( (T/2) - 1 ) + \mathcal{O}(T)
    \nonumber\\&
    \stackrel{(b)}{=} \frac{1}{2} \log \text{det}\big( \bar{\bar{\fA}}_{\mathbb{G}_{T^*}} \big) + \left( \frac{2 - ( \kappa + 4l + b )}{4} \right) T \log T
    \nonumber\\&
    + T \log  \Big( \frac{C_{\rm avg}}{\sqrt{ae}} \Big) + \mathcal{O}(T) ,
    \label{Eq.Log_M}
\end{align}
where $(a)$ follows by $\floor{T/2} > (T/2) - 1$ and $(b)$ holds since $\log(t-1) \geq \log t - 1$ for $t \geq 2$ and $\floor{T/2} \leq (T/2).$ Next, in order to determine the scaling behavior of terms in the right hand side of \eqref{Eq.Log_M}, we present a useful lemma.
\begin{customlemma}{2}[\textbf{Upper Bound on the Determinant of a Matrix}]
\label{Lem.UB_Determinant}
Let $\fX$ be a square matrix with real entries, dimension $I \times I$ whose $i$-column vector is $x_i$ for every $i \in [\![I]\!].$ Then, it holds $| \text{det}(\fX) | \leq \prod_{i=1}^I \norm{x_i}.$
\end{customlemma}
\begin{proof}
    The proof is provided in \cite[Ch.~14]{Garling07} .
\end{proof}
Next, employing Lemma~\ref{Lem.UB_Determinant} with setting $I=T, \fx_i = \bar{\bar{\fA}}_{\mathbb{G}_{T^*}}^{n_t} $ and condition \textbf{C4} provided in Subsection~\ref{Subsec.Main_Results} we calculate the scaling order for the determinant of matrix $\bar{\bar{\fA}}_{\mathbb{G}_{T^*}}:$
\begin{align}
    \label{Ineq.Scaling_Function_Achievability}
    & \log \big( \text{det}\big( \bar{\bar{\fA}}_{\mathbb{G}_{T^*}} \big) \big) \leq \log \Big( \prod_{t=1}^T \| \bar{\bar{\fA}}_{\mathbb{G}_{T^*}}^{n_t} \| \Big) 
    \nonumber\\&
    = \sum_{t=1}^T \log \| \bar{\bar{\fA}}_{\mathbb{G}_{T^*}}^{n_t} \| \stackrel{(a)}{\leq} T \log o(T) = o(T \log T) .
\end{align}
where $(a)$ uses condition \textbf{C4}, cf. Subsection~\ref{Subsec.Main_Results}.

Therefore, recalling \eqref{Eq.Log_M}, we obtain
\begin{align}
    \log M & \geq \left( \frac{2 - ( \kappa + 4l + b )}{4} \right) T \log T + T \log \Big( \frac{C_{\rm avg}}{\sqrt{ae}} \Big)
    \nonumber\\&
    + o(T \log T) + \mathcal{O}(T) .
    \label{Eq.Log_M_2}
\end{align}
Thereby, the dominant term in \eqref{Eq.Log_M_2} is of order $T \log T.$ Hence, for obtaining a finite value for the lower bound of the rate, $R,$ \eqref{Eq.Log_M_2} induces the scaling law of $M$ to be $2^{(T\log T)R}.$ Therefore, we obtain
\begin{align}
    R & \geq \frac{1}{T \log T} \Bigg[ \left( \frac{2 - (\kappa + 4l + b)}{4} \right) T \log T
    \nonumber\\&\hspace{15mm}
    + T \log \Big( \frac{C_{\rm avg}}{\sqrt{ae}} \Big) + o(T \log T) + \mathcal{O}(T) \Bigg] ,
  \end{align}
which tends to $(2 - (\kappa + 4l))/4$ when $T \to \infty$ and $b \rightarrow 0.$

\section{Decoding Conventions}
\label{App.Decoding_Conventions}
In the following, for the sake of conducting concise derivations related to the error analysis, we introduce a number of conventions:
\begin{itemize}[leftmargin=*]
    \item $\bar{c}_k^i \triangleq \sum_{n=1}^N a_{k,n} v_n c_n^i$ is the affine symbol.
    \item $Y_k(i) \sim \text{Pois} (\bar{c}_k^i + \lambda_k )$ is the channel output at receptor $k$ given that $\fx = \fc_n^i .$
    \item The output vector is defined as $\fY(i) = (Y_1(i),\ldots, Y_{K}(i)) .$
    \item $k_{\text{max},1} \triangleq \underset{k \in \mathbb{E}_T}{\argmax}\, \text{Var}[ ( Y_k(i) - ( \bar{c}_k^i + \lambda_k ) )^2].$
    \item $k_{\text{max},2} \triangleq \underset{k \in \mathbb{E}_T}{\argmax}\, \text{Var}[Y_k(i)].$
    \item $A_k \triangleq F_k \cdot \left( A_{\text{max}} v_{\text{max}} C_{\text{avg}} \hspace{-.4mm} + \hspace{-.4mm} \lambda_{\text{max}} \right) = \mathcal{O}(T^l) ,\; \forall k \in [\![K]\!].$
    \item $U_{k,C_{\rm avg}} \triangleq A_k^4 + A_k^3 + A_k^2 + A_k ,\; \forall k \in [\![K]\!].$
    \item Let $\mathbb{E}_T \triangleq \allowbreak \{k_1,\ldots,k_T\}$ denotes the set of $T$ row indices whose corresponding rows are linearly independent with $k_1 < k_2 < \allowbreak \ldots \allowbreak < k_T.$
\end{itemize}

\section{Extensive Error Analysis}
\label{App.Extensive_Error_Analysis}

\textbf{\textcolor{mycolor12}{Calculation of Type I Error:}}
The type I errors occur when the transmitter sends $\fc^i,$ yet $\fY\notin\Z_i.$ For every $i \in [\![M]\!],$ the type I error probability is bounded by 
\begin{align*}
    P_{e,1}(i) = \Pr\big( \fY(i) \in \Z_i^c \big) = \Pr\big( \big| Z(\fY(i),\fc^i) \big| > \psi_T \big) . 
\end{align*}
To bound $P_{e,1}(i),$ we apply Chebyshev's inequality, namely
\begin{equation*}
\vspace{-1mm} \resizebox{.49\textwidth}{!}{
  $\begin{aligned}
    \Pr\big(\big| Z(\fY(i); \fc^i) \hspace{-.6mm} - \hspace{-.6mm}\mathbb{E} \big[ Z(\fY(i); \fc^i) \big] \big| \hspace{-.6mm} > \hspace{-.6mm} \psi_T \big) \hspace{-.6mm} \leq \hspace{-.6mm} \frac{\text{Var} \big[ Z(\fY(i);\fc^i) \big]}{\psi_T^2} .
  \end{aligned}$}
\end{equation*}
Note $\mathbb{E}\big[ Z(\fY(i); \fc^i) \big] = 0.$ Next, exploiting the property of the statistical independence of observations and $\text{Var} \big[ Y_k(i) \big] = \bar{c}_k^i + \lambda_k \geq 0,$ variance of expression in \eqref{Eq.DM} is obtained as
\begin{equation*}
\vspace{-2mm} \resizebox{.45\textwidth}{!}{
  $\begin{aligned}
    \text{Var}\big[ Z(\fY(i);\fc^i ) \big]
    \hspace{-.6mm} \leq \hspace{-.6mm} T^{-2} \hspace{-1mm} \sum_{k \in \mathbb{E}_T} \text{Var}[ ( Y_k(i) \hspace{-.6mm} - \hspace{-.6mm} ( \bar{c}_k^i + \lambda_k ) )^2 ] .\hspace{-3mm}
    \vspace{-4mm}
    \end{aligned}$}
\end{equation*}
To establish an upper bound on the sum in above inequality, we present a helpful lemma.

\begin{customlemma}{3}[\textbf{Poisson Moment's Upper Bound}]
\label{Lem.MGF}
Let $X \sim \text{Pois}(\lambda_X)$ be a Poisson RV with mean $\lambda_X.$ Then, it holds $\mathbb{E}[ (X -\lambda_X)^4 ] \leq 7 ( \lambda_X^4 + \lambda_X^3 + \lambda_X^2 + \lambda_X ) .$
\end{customlemma}
\begin{proof}
    The moment-generating function (MGF) of a Poisson variable $Z\sim\text{Pois}(\lambda_X)$ is $G_Z(\phi) = e^{\lambda_X(e^{\phi}-1)}.$ Hence, for $X=Z-\lambda_X,$ the MGF is given by $G_X(\phi) = e^{\lambda_X (e^{\phi} - 1 - \phi)}.$ Since the fourth non-central moment equals the fourth order derivative of the MFG at $\phi = 0,$ we have
\begin{align*}
    & \mathbb{E}\{ X^4 \} =
    \frac{d^4}{d\phi^4}G_X(\phi)
    \Big|_{\phi=0}
    \nonumber\\&
    = \lambda_X \big( \lambda_X^3 e^{3\phi} + 6\lambda_X^2 e^{2\phi} + 7\lambda_X e^{\phi} + 1 \big) e^{\phi + \lambda_X e^{\phi} - \lambda_X}
    \big|_{\phi = 0}
    \nonumber\\&
    = \lambda_X^4 + 6 \lambda_X^3 + 7\lambda_X^2 + \lambda_X \leq 7
    \big( \lambda_X^4 + \lambda_X^3 + \lambda_X^2 + \lambda_X \big) .
\end{align*}
\end{proof}

Next, we exploit Lemma~\ref{Lem.MGF} and obtain the following
\begin{align}
   \text{Var}\big[ Z(\fY(i);\fc^i) \big]
   & \stackrel{(a)}{\leq} \frac{T \cdot \mathbb{E} \big[ \big( Y_{k_{\text{max},1}}(i) - \big( \bar{c}_{k_{\text{max},1}}^i + \lambda_{k_{\text{max},1}} \big) \big)^4 \big]}{T^{2}}
   \nonumber\\&
   \stackrel{(b)}{\leq} \frac{7U_{k,C_{\rm avg}}}{T} ,
   \label{Ineq.4thMoment}
\end{align}
where inequality $(a)$ follows since $\text{Var}[Z] \leq \mathbb{E}[Z^2],$ with $Z = \big( Y_{k_{\text{max},1}}(i) - \big( \bar{c}_{k_{\text{max},1}}^i + \lambda_{k_{\text{max},1}} \big) \big)^2 $ and $(b)$ holds by Lemma~\ref{Lem.MGF}.

Thus, we can bound the type I error probability as follows
\begin{align}
    P_{e,1}(i) \leq \frac{7U_{k,C_{\text{avg}}}}{T \psi_T^2} \triangleq \zeta_1 = \frac{\mathcal{O}(T^{4l})}{T \psi_T^2} = \frac{\mathcal{O}(1)}{T^{b}} \leq \varepsilon_1 ,
\end{align}
for sufficiently large $T$ and arbitrarily small $\varepsilon_1 > 0.$

\textbf{\textcolor{mycolor12}{Calculation of Type II Error:}}
We examine type II errors, i.e., when $\fY\in\Z_j$ while the transmitter sent $\fc^i$ with $i \neq j.$ Then, for every $i,j \in [\![M]\!],$ the type II error probability reads
\begin{align}
    P_{e,2}(i,j) = \Pr \big( \big| Z(\fY(i);\fc^j) \big| \leq \psi_T \big) ,
    \label{Eq.Pe2G_Extensive}
\end{align}
where $Z(\fY(i);\fc^j) = \omega - \phi ,$ with $\phi \triangleq T^{-1} \sum_{k \in \mathbb{E}_T} Y_k(i)$ and $\omega \triangleq T^{-1} \sum_{k \in \mathbb{E}_T} ( \omega_i + \omega_{i,j} )^2$ with $\omega_i \hspace{-.4mm} \triangleq \hspace{-.4mm} Y_k(i) - ( \hspace{-.5mm} \bar{c}_k^i + \lambda_k )$ and 
$\omega_{i,j} \triangleq \bar{c}_k^i - \bar{c}_k^j.$ 

Observe that $\omega = \omega_1 + \omega_2$ where $$\omega_1 \triangleq T^{-1} \sum_{k \in \mathbb{E}_T} \omega_i^2 + \omega_{i,j}^2\quad \text{ and } \quad \omega_2 \triangleq  2T^{-1} \sum_{k \in \mathbb{E}_T} \omega_i \omega_{i,j}.$$

Next, let $\Omega_0 = \{ | \omega_2 | > \psi_T \}$ and $\Omega_1 = \{ \omega_1 - \phi \leq 2\psi_T \}.$ Exploiting the reverse triangle inequality, i.e., $|\omega| - |\phi| \leq |\omega - \phi|,$ we get the following bound on the type II error probability
\begin{align}
    P_{e,2}(i,j) & = \Pr\big( |\omega - \phi| \leq \psi_T \big)
    \nonumber\\&
    \leq \Pr\big( |\omega| - |\phi| \leq \psi_T \big)
    \nonumber\\&
    \stackrel{(a)}{=} \Pr\big( \omega - \phi \leq \psi_T \big) ,
\end{align}
where $(a)$ follows since $\phi \geq 0$ and $\omega \geq0.$ Now, utilizing the principle of total probability on the event $\D = \big\{ \omega - \phi \leq \psi_T \big\}$ over $\Omega_0$ and its complement $\Omega_0^c \,,$ we obtain
\begin{align*}
    P_{e,2}(i,j) \stackrel{(a)}{\leq} \Pr\left(\Omega_0 \right) + \Pr \left( \D \cap \Omega_0^c \right) 
    \stackrel{(b)}{\leq} \Pr\left(\Omega_0 \right) + \Pr\left(\Omega_1 \right),
\end{align*}
where $(a)$ holds by $\D \cap \Omega_0 \subset \Omega_0$ and $(b)$ follows by $\Pr \left( \D \cap \Omega_0^c \right) \leq \Pr\left(\Omega_1 \right),$ which holds by the followings:
\begin{align}
   \Pr \left( \D \cap \Omega_0^c \right) & = \Pr \big( \big\{ \omega - \phi \leq \psi_T \big\}
   \cap \big\{| \omega_2 | \leq \psi_T \, \big\} \big)
   \nonumber\\&
   \stackrel{(a)}{\leq} \Pr \big( \big\{ \omega_1 - \phi \leq 2 \psi_T \big\} \big)
   \nonumber\\&
   = \Pr \left(\Omega_1 \right) ,
\end{align}
where $(a)$ holds since $\psi_T-\omega_2\leq 2\psi_T$ given $| \omega_2 | \leq \psi_T.$

We now proceed with bounding $\Pr\left(\Omega_0 \right).$ By Chebyshev's inequality, the probability of this event can be bounded by
\begin{align}
    & \Pr(\Omega_0) \leq \frac{\text{Var}\big[ \sum_{k \in \mathbb{E}_T} \big( \bar{c}_k^i - \bar{c}_k^j \big) \big( Y_k(i) - \big( \bar{c}_k^i + \lambda_k \big) \big) \big]}{{T^2 \psi_T^2}/{4}}
    \nonumber\\&
    \stackrel{(a)}{=} \frac{\sum_{k \in \mathbb{E}_T} \text{Var} \big[ \big( \bar{c}_k^i - \bar{c}_k^j \big) \big( Y_k(i) - ( \bar{c}_k^i + \lambda_k ) \big) \big]}{{T^2 \psi_T^2}/{4}}
    \nonumber\\&
    \stackrel{(b)}{\leq} \frac{4 (F_k \cdot A_{\text{max}} v_{\text{max}})^2 A_k \big\| \fc^i - \fc^j \big\|_{\infty}^2 }{ T \psi_T^2} ,
    \label{Ineq.E_0_1}
\end{align}
where $(a)$ follows since channel outputs $Y_k(i)$ for different receptors are independent, and $(b)$ exploits the followings:
\begin{itemize}[leftmargin=*]
    \item $\underset{n \in [\![N]\!]}{\max}\, | c_n^i - c_n^j | \triangleq \| \fc^i - \fc^j \|_{\infty}.$
    \item $| \bar{c}_{k_{\text{max},2}}^i - \bar{c}_{k_{\text{max},2}}^j | \leq F_k \cdot A_{\text{max}} v_{\text{max}} \cdot \| \fc^i - \fc^j \|_{\infty}.$
    \item $\text{Var}[ Y_{k_{\text{max},2}}(i) ] = \bar{c}_{k_{\text{max},2}}^i + \lambda_{k_{\text{max},2}} \leq A_k .$
\end{itemize}
Next, observe that
\begin{align*}
    \| \fc^i - \fc^j \|_{\infty}^2 \hspace{-.6mm} \stackrel{(a)}{\leq} \hspace{-.6mm} ( \| \fc^i \|_{\infty} + \| \fc^j \|_{\infty} )^2 \hspace{-.6mm} \stackrel{(b)}{\leq} \hspace{-.6mm} ( C_{\text{avg}} + C_{\text{avg}} )^2 = 4C_{\text{avg}}^2 ,
\end{align*}
where $(a)$ holds by the triangle inequality and $(b)$ is valid by $\| \fc^i \|_{\infty} \leq C_{\text{avg}}.$ Thereby, we obtain the below bound on $\Omega_0:$
\begin{align}
    \label{Ineq.E_0_2_Extensive}
    \Pr(\Omega_0) \leq \frac{16 C_{\text{avg}}^2 (F_k \cdot A_{\text{max}} v_{\text{max}})^2 A_k }{T \psi_T^2} \triangleq \zeta_0 \stackrel{(a)} = \frac{\mathcal{O}(1)}{T^{l + b}} ,
\end{align}
for sufficiently large $T,$ where $\zeta_0 > 0$ is an arbitrarily small constant and $(a)$ exploits condition \textbf{C2} provided in Subsection~\ref{Subsec.Main_Results}, i.e., $F_k = \mathcal{O}(T^l), \, \forall k \in [\![K]\!].$

We now proceed with bounding $\Pr\left(\Omega_1 \right)$ as follows. To this end, we use the following lemma.
\newpage

\begin{customlemma}{4}[\textbf{Minimum Distance of Reduced Affine Codebook}]
\label{Lem.Minimum_Distance_Affine}
Let $\mathbb{E}_T \subset [\![K]\!]$ with $|\mathbb{E}_T| = T,$ be a subset of $T$ row indices of matrix $\bar{\fA}$ for which the corresponding rows are linearly independent and $\bar{\fA}_{\mathbb{E}_T}$ be a sub-matrix of $\bar{\fA} = \fA\fV$ containing $T$ of its rows whose indices are in $\mathbb{E}_T.$ Moreover, let $\bar{\fA} = \fU\fSigma\fW^\dag$ be the SVD of $\bar{\fA}$ and $\fU_T$ denote the sub-matrix of $\fU$ containing the first $T$ columns and the $T$ rows whose indices are in $\mathbb{E}_T.$ Then, the minimum distance of reduced affine codebook is given by
\begin{align}
    d_{\rm min}(\bar{\bar{\mathbb{C}}}_0^{\fA}) \triangleq \underset{\substack{i,j \in [\![M]\!] \\ i \neq j}}{\min} \Big[ \sum_{k \in \mathbb{E}_T} \big( \bar{c}_{k}^i - \bar{c}_{k}^j \big)^2 \Big]
    \geq 4 T \epsilon_T . 
\end{align}
\end{customlemma}
\begin{proof}
The proof is provided in Appendix~\ref{App.Minimum_Distance_Reduced_Affine_Codebook}.
\end{proof}
\vspace{-1mm}
Note that the codebook construction in Subsection~\ref{Subsec.Achievability} implies that each reduced affine codeword is surrounded by a hyper spheres of radius $\sqrt{T\epsilon_T}.$ Now, exploiting Lemma~\ref{Lem.Minimum_Distance_Affine} establishes a minimum distance for $\bar{\bar{\mathbb{C}}}_0^{\fA},$ i.e.,
\begin{align}
    \label{Ineq.LB_Codewords_Type_II}
    \sum_{k \in \mathbb{E}_T} \big( \bar{c}_{k}^i - \bar{c}_{k}^j \big)^2 \geq 4 T \epsilon_T .
\end{align}
\vspace{-1mm}
Thus, we can establish the below bound for the event $\Omega_1$:
\begin{align}
    \label{Ineq.E_1}
    & \Pr (\Omega_1) = \Pr \big( \omega_1 - \phi \leq 2\psi_T \big)
    \nonumber\\& 
    \stackrel{(a)}{\leq} \Pr \Big( \hspace{-.6mm} \sum_{k \in \mathbb{E}_T} \hspace{-.6mm} \big( Y_k(i) - ( \bar{c}_k^i + \lambda_k ) \big)^2 - Y_k(i) \hspace{-.6mm} \leq \hspace{-.6mm} - d_{\rm min}(\bar{\bar{\mathbb{C}}}_0^{\fA}) / 4 \Big)
    \nonumber\\&
    \stackrel{(b)}{\leq} \frac{\text{Var} \big[ \sum_{k \in \mathbb{E}_T} \big( Y_k(i) - \left( \bar{c}_k^i + \lambda_k \right) \big)^2 - Y_k(i) \big]}{T \psi_T^2}
    \nonumber\\&
    \stackrel{(c)}{\leq} \frac{7U_{k,C_{\rm avg}}}{T\psi_T^2} \triangleq \zeta_1 = \frac{\mathcal{O}(T^{4l})}{T \psi_T^2} = \frac{\mathcal{O}(1)}{T^{b}} ,
 \end{align}
for sufficiently large $T,$ where $\zeta_1 > 0$ is an arbitrarily small constant. Here, 
\begin{itemize}[leftmargin=*]
    \item $(a)$ follows from \eqref{Ineq.LB_Codewords_Type_II} and the decoding threshold $\psi_T = 4a / 3T^{( 2 - ( \kappa + 4l + b)) / 2},$
    \item $(b)$ holds by the Chebyshev's inequality, and exploiting decoding threshold $\psi_T,$
    \item $(c)$ holds by Lemma~\ref{Lem.MGF} and similar arguments in the type I error analysis provided in \eqref{Ineq.4thMoment}.
\end{itemize}
Hence, we obtain the following upper bound on the type II error probability

\begin{align}
    P_{e,2}(i,j) \leq \Pr(\Omega_0) + \Pr(\Omega_1) \leq \zeta_0 + \zeta_1 \leq \varepsilon_2 .
\end{align}

\renewcommand{\thesectiondis}[2]{\Alph{section}:}
\section{Minimum Distance of Reduced Affine Codebook}
\label{App.Minimum_Distance_Reduced_Affine_Codebook}
\vspace{0mm}
Let define the following conventions:
\begin{itemize}[leftmargin=*]
	\item $\bar{\fA} = \fU \mathbf{\Sigma} \fW^\dag$ is the SVD of $\bar{\fA}.$
    \item $\fDelta_{i,j} \triangleq \mathbf{\Sigma} \fW^\dag(\fc_i - \fc_j).$
    \item $\fDelta'_{i,j}$ is a sub-vector of $\fDelta_{i,j}$ consisting first $T$ elements.
    \item Let $\fU^1$ be the sub-matrix of $\fU$ containing the rows whose indices are in $\mathbb{E}_T$ and $\fU^2$ be the sub-matrix of $\fU$ containing the remaining rows. Now, observe that left multiplication by an appropriate permutation matrix $\fP_r$ of size $K \times K$ can rearrange the corresponding $T$ rows. Therefore, matrix $\fU$ after multiplication by $\fP_r$ can be partitioned into two sub-matrices $\fU^1$ and $\fU^2,$ i.e., $\fP_r \fU = [\fU_{T \times K}^1 \fU_{K-T \times K}^2]^T.$
    \item $\sum_{k \in \mathbb{E}_T} \big( \bar{c}_{k}^i - \bar{c}_{k}^j \big)^2 = \| \fU^1 \fDelta_{i,j} \|^2.$
\end{itemize}
Observe that the distance between the codewords in $\bar{\bar{\mathbb{C}}}_0^{\fA}$ reads
\begin{align}
	\label{Eq.Decomposition}
    & \big\| \bar{\bar{\fc}}_{\fA}^i - \bar{\bar{\fc}}_{\fA}^j \big\|^2 = \norm{\bar{\fc}_i - \bar{\fc}_j}^2 = \| \bar{\fA} (\fc_i - \fc_j)  \|^2 \stackrel{(a)}{=} \|\fU \fDelta_{i,j} \|^2 
    \nonumber\\&
    \stackrel{(b)}{=} \| \fP_r\fU \fDelta_{i,j}  \|^2 =  \| \fU^1 \fDelta_{i,j}  \|^2 +  \| \fU^2 \fDelta_{i,j} \|^2 \geq 4 T \epsilon_T ,
\end{align}
where $(a)$ exploits the SVD of $\bar{\fA}$ and definition of $\fDelta_{i,j}$ and $(b)$ holds since multiplication with permutation matrix is invariant under Euclidean norm.

Next, we proceed to establish a lower bound for $\| \fU^1 \fDelta_{i,j} \|^2$ in \eqref{Eq.Decomposition}. Consider that only the first $T$ elements of $\fDelta_{i,j}$ are non-zero, which we denote by $\fDelta'_{i,j}.$ Now, since the matrix $\fU$ is unitary, it is full rank, therefore, a subset of such row vectors, e.g., the set of row vectors of $\fU_1$ are linearly independent. This implies that $\fU_1$ is a full row rank matrix.

Next, recalling the definition of the rank for a rectangular matrix, i.e., the size of largest square full rank sub-matrix, we conclude that there exists at least one square full rank sub-matrix of $\fU_1$ denoted by $\fU_T'$ consisting of $T$ rows and $T$ columns. Observe that without loss of generality, we may assume that the $T$ columns of $\fU_T'$ are listed in the first $T$ columns of $\fU_1,$ since otherwise right multiplication of $\fU_1$ by an appropriate permutation matrix $\fP_c$ of size $K \times K$ rearranges the corresponding columns to the first $T$ indices. Thereby, we denote the square sub-matrix constructed by $T$ rows and $T$ permuted columns by $\fU_T.$ Therefore, we obtain
\begin{align}
    \label{Ineq.U1_LB-0}
	\| \fU^1 \fDelta_{i,j} \|^2 = \| \fU_T \fDelta_{i,j} \|^2 = \| \fU_T \fDelta'_{i,j} \|^2 ,
\end{align}
where for the second equality we used $\| \fDelta_{i,j} \| = \| \fDelta'_{i,j} \|.$ Next, in order to establish a lower bound for $\| \fU_T \fDelta'_{i,j} \|$ given in \eqref{Ineq.U1_LB-0} which is the norm of product of a square full rank matrix and a vector, we present a useful lemma.
\begin{customlemma}{5}[\textbf{Bounds on the Norm of Matrix-Vector Product}]
\label{Lem.Matrix_Vector_Product_Bounds}
Let $\fB_{h \times h}$ be a full rank square matrix with $\sigma_{\min}(\fB)$ being the minimum singular values of $\fB.$ Then, for every vector $\fx \in \mathbb{R}^h$ we have $\norm{\fB\fx} \geq \sigma_{\min}(\fB) \cdot \norm{\fx}.$
\end{customlemma}
\begin{proof}
The proof is provided in Appendix~\ref{App.Matrix_Vector_Product_Bounds}.
\end{proof}
Hence, we get the following lower bound on the distance between reduced affine codewords:
\begin{align}
    \label{Ineq.U1_LB}
    \| \fU_T \fDelta'_{i,j} \|^2 & \hspace{-.6mm} \stackrel{(a)}{\geq} \hspace{-.6mm} \sigma^2_{\min}(\fU_T) \cdot \| \fDelta'_{i,j} \|^2
    \hspace{-.6mm} \stackrel{(b)}{\geq} \hspace{-.6mm} 4 \sigma^2_{\min}(\fU_T) \cdot T \epsilon_T , \hspace{-1mm}
\end{align}
where $(a)$ employs Lemma~\ref{Lem.Matrix_Vector_Product_Bounds} with setting $\fB = \fU_T, \fx = \fDelta'_{i,j}$ and $(b)$ uses $\| \fDelta'_{i,j} \| = \| \fDelta_{i,j} \| = \| \fU \fDelta_{i,j} \|$ and \eqref{Eq.Decomposition}.

Next, to determine whether the smallest singular values of $\fU_T$ scales with parameter $T,$ or not, we present a lemma.

\begin{customlemma}{6}[\textbf{Singular Values of Sub-Matrix of Unitary Matrix}]
\label{Lem.SV_Sub_Matrix_Unitary_Matrix}
Let $\fU$ be an $I \times I$ unitary matrix and let $\fU_s$ be its $p \times q$ sub-matrix. Then $\fU$ has all singular values less than or equal to one, and the number of singular values less than one (counting also zero singular values) does not exceed $I - \max(p,q).$
\end{customlemma}
\begin{proof}
    The proof is provided in \cite[Ch.~2]{Fiedler09} .
\end{proof}
Next, employing Lemma~\ref{Lem.SV_Sub_Matrix_Unitary_Matrix} for unitary matrix $\fU_T$ with substitution of $I = K$ and $p=q=T,$ we conclude that all of the singular values of $\fU_T$ must be one. Thereby, recalling \eqref{Ineq.U1_LB} and exploiting Lemma~\ref{Lem.SV_Sub_Matrix_Unitary_Matrix}, we obtain
\begin{align}
    \| \fU^1 \fDelta_{i,j} \|^2 \geq 4 \sigma^2_{\min} T \epsilon_T = 4 T \epsilon_T .
\end{align}

\renewcommand{\thesectiondis}[2]{\Alph{section}:}
\section{Proof of Lemma~\ref{Lem.Matrix_Vector_Product_Bounds}}
\label{App.Matrix_Vector_Product_Bounds}
The proof follows by using the SVD for matrix $\fB.$ Let $\fB = \fU \mathbf{\Sigma} \fW^\dag$ be the SVD of $\fB.$ Then,
\begin{align}
    & \norm{\fB\fx} = \| \fU \mathbf{\Sigma} \fW^\dag \fx \| \stackrel{(a)}{=} \| \mathbf{\Sigma} \fW^\dag \fx \|
    \nonumber\\&
    \stackrel{(b)}{=} \sqrt{\sum_{t=1}^h \sigma_t^2 |(\fW^{\dag}\fx)_t|^2} \stackrel{(c)}{\geq} \sqrt{\sigma_{\min}(\fB)^2 \cdot \sum_{t=1}^h \sigma_t^2 |(\fW^{\dag}\fx)_t|^2}
    \nonumber\\&
    = \sigma_{\min}(\fB) \cdot \| \fW^\dag \fx \|^2 \stackrel{(d)}{=} \sigma_{\min}(\fB) \cdot \norm{\fx} ,
\end{align}
where $(a)$ holds since matrix multiplication is invariant under unitary matrix $\fU,$ $(b)$ exploits definition of Euclidean norm, $(c)$ holds since $\sigma_{\min}(\fB) \leq \sigma_t \,, \forall t \in [\![h]\!],$ and $(d)$ follows since matrix multiplication is invariant under unitary matrix $\fW^{\dag}.$

\renewcommand{\thesectiondis}[2]{\Alph{section}:}
\section{Proof of Lemma~\ref{Lem.Converse}}
\label{App.Converse_Lemma}
In the following, we provide the proof of Lemma~\ref{Lem.Converse}. The method of proof is by contradiction, namely we assume that the condition given in Lemma~\ref{Lem.Converse} is violated and then we show that this leads to a contradiction, i.e., sum of the type I and type II error probabilities converge to one, i.e., $\lim_{N \to \infty} \left[ P_{e,1}(i_1) + P_{e,2}(i_2,i_1) \right] = 1.$

Fix $\varepsilon_1,\varepsilon_2 > 0.$ Let $\eta > 0$ be arbitrarily small constant and $\phi_K \triangleq (K T^{-(l + b)})^{-1/2}$ where $b$ is an arbitrary small constant. Assume to the contrary that there exist two messages $i_1$ and $i_2,$ where $i_1\neq i_2,$ meeting the error constraints in Definition~\ref{Def.Affine_Identification_Code}, such that for all $k \in [\![K]\!],$ we have
\begin{align}
    \label{Ineq.Converse_Lem_Complement}
    \Big| 1 - \frac{d_k^{i_2}}{d_k^{i_1}} \Big| \leq \theta_T .
\end{align}
In order to establish a contradiction, we would bound the sum of the two error probabilities, $P_{e,1}(i_1) + P_{e,2}(i_2,i_1),$ from below. To this end, let $\mathscr{F}_{i_1} = \{\fy \in \Z_{i_1} \;:\, K^{-1} \sum_{k = 1}^K y_k \leq \bar{A}_{\rm max} + \phi_K \},$ where $\Z_{i_1} \subseteq \mathbb{N}_0^K$ is the decoding set for message $i_1.$ Next, observe that sum of the errors can be bounded by
    \begin{align}
    \label{Eq.Error_Sum_1}
    & P_{e,1}(i_1) + P_{e,2}(i_2,i_1)
    \nonumber\\&
    = 1 - \sum_{\fy\in\Z_{i_1}} W^{K} \big( \fy \, \big| \, \fc^{i_1} \big) + \sum_{\fy \in \Z_{i_1}} W^{K} \big( \fy \, \big| \, \fc^{i_2} \big)
    \nonumber\\&
    \geq 1 - \sum_{\fy \in \Z_{i_1}} W^{K} \big( \fy \, \big| \, \fc^{i_1} \big) + \sum_{\fy \in \Z_{i_1} \cap \mathscr{F}_{i_1}} W^{K} \big( \fy \, \big| \, \fc^{i_2} \big) .
    \end{align}
    
    Now, consider the first sum over $\Z_{i_1}$ in \eqref{Eq.Error_Sum_1},
    \begin{align}
        \label{Eq.LTP}
        & \sum_{\Z_{i_1}} W^{K} \big( \fy \, \big| \, \fc^{i_1} \big)
        \nonumber\\&
        = \sum_{\Z_{i_1} \cap \mathscr{F}_{i_1}} W^{K} \big( \fy \, \big| \, \fc^{i_1} \big) + \sum_{\Z_{i_1} \cap \mathscr{F}_{i_1}^c} W^{K} \big( \fy \, \big| \, \fc^{i_1} \big) ,
    \end{align}
    where the equality follows from applying the law of total probability on $\Z_{i_1}$ with respect to $(\mathscr{F}_{i_1},\mathscr{F}_{i_1}^c).$ Next, we provide an upper bound for the second sum in \eqref{Eq.LTP} as follows
    \begin{align}
        \label{Eq.LTP2}
        & \sum_{\Z_{i_1} \cap \mathscr{F}_{i_1}^c} W^{K} \big( \fy \, \big| \, \fc^{i_1} \big) = \Pr(\Z_{i_1} \cap \mathscr{F}_{i_1}^c)
        \nonumber\\&
        \leq \Pr( \mathscr{F}_{i_1}^c) = \Pr\big( K^{-1} \sum_{k=1}^K Y_k(i_1) > \bar{A}_{\rm max} + \phi_K \big) .
    \end{align}
    Next, let define $\xi_K(i_1) \triangleq \bar{A}_{\rm max} + \phi_K - \mathbb{E} [ \bar{\fY}(i_1) ].$ Then, the application of Chebyshev's inequality to the probability term in \eqref{Eq.LTP2} yields the following upper bound
    \begin{align}
        & \Pr \big( \bar{\fY}(i_1) - \mathbb{E} [ \bar{\fY}(i_1) ] > \xi_K(i_1) \big)
        \nonumber\\&
        \stackrel{(a)}{\leq} \frac{\text{Var} [ \bar{\fY}(i_1) ]}{\xi_K^2(i_1)} \stackrel{(b)}{\leq} \frac{\bar{A}_{\rm max}}{K \phi_K^2} = \mathcal{O}(T^{-b}) \leq \eta ,
        \label{Ineq.ErrorI_Complement}
   \end{align}
for sufficiently large $T,$ where $(a)$ exploits Chebyshev's inequality and $(b)$ uses $\text{Var} [Y_k(i_1)] = \mathbb{E} [Y_k(i_1)] \leq \bar{A}_{\rm max} , \forall \, k \in [\![K]\!].$

Next, recalling the sum of error probabilities in \eqref{Eq.Error_Sum_1}, exploiting the bound \eqref{Ineq.ErrorI_Complement} leads to
\begin{align}
    \label{Eq.Error_Sum_2}
    & P_{e,1}(i_1)+P_{e,2}(i_2,i_1)
    \nonumber\\&
    \geq 1 - \hspace{-1mm} \sum_{\Z_{i_1} \cap \mathscr{F}_{i_1}} \hspace{-2mm} \big[ W^{K} \big( \fy \, \big| \, \fc^{i_1} \big) - W^{K} \big( \fy \, \big| \, \fc^{i_2} \big) \big] - \eta .
\end{align}
Next, employing $d_k^{i_2} - d_k^{i_1} \leq \theta_T d_k^{i_1}$ and $1 - d_k^{i_2} / d_k^{i_1} \leq \theta_T$ from \eqref{Ineq.Converse_Lem_Complement}, we decompose the square brackets in \eqref{Eq.Error_Sum_2}:
    \begin{align}
        \label{Ineq.Cond_Channel_Diff}
        & W^{K} \big( \fy \, \big| \, \fc^{i_1} \big) - W^{K} \big( \fy \, \big| \, \fc^{i_2} \big)
        \nonumber\\&
        = W^{K} \big( \fy \, \big| \, \fc^{i_1} \big) \cdot \Big[ 1 - \prod_{k=1}^{K} e^{-\theta_T d_k^{i_1}} \left(1 -\theta_T\right)^{y_k} \Big] ,
    \end{align}

    Now, we bound the product term inside the bracket
    \begin{align}
        \label{Ineq.Converse_Product_LB}
        & \prod_{k=1}^{K} \hspace{-.3mm} e^{-\theta_T d_k^{i_1}}\left( 1 - \theta_T \right)^{y_k} = e^{- \theta_T \sum_{k=1}^{K} d_k^{i_1}} \hspace{-.6mm} \cdot \hspace{-.4mm} \left( 1 - \theta_T \right)^{\sum_{k=1}^{K} y_k}
        \nonumber\\&
        \overset{(a)}{\geq} e^{-K \theta_T \bar{A}_{\rm max}} \cdot \left( 1 - \theta_T \right)^{K (\bar{A}_{\rm max} + \phi_K)}
        \nonumber\\&
        = e^{K \theta_T \phi_K } \cdot e^{-K \theta_T (\bar{A}_{\rm max} + \phi_K)} \cdot \left( 1 - \theta_T \right)^{K (\bar{A}_{\rm max} + \phi_K)}
        \nonumber\\&
        \overset{(b)}{\geq} e^{-K \theta_T (\bar{A}_{\rm max} + \phi_K)} \cdot \left( 1 - \theta_T \right)^{K (\bar{A}_{\rm max} + \phi_K)}
        \nonumber\\&
        \overset{(c)}{\geq} e^{-K \theta_T (\bar{A}_{\rm max} + \phi_K)} \cdot \left( 1 - K\theta_T \right)^{\bar{A}_{\rm max} + \phi_K}
        \nonumber\\&
        \stackrel{(d)}{=} f(K\theta_T)
        \stackrel{(e)}{\geq} 1 - 3 ( \bar{A}_{\rm max} + \phi_K ) K \theta_T
        \nonumber\\&
        \geq 1 - \mathcal{O}( \bar{A}_{\rm max} K \theta_K )
        = 1 - \mathcal{O}\big( KT^l / T^{\kappa(1-l\delta_{\kappa,1}) + 2l + b} \big)
        \nonumber\\&
        \stackrel{(f)}{=}
        \begin{cases}
            1 - \mathcal{O}(K^{1 - \tau(\kappa+l+b)}) & 0 < \kappa < 1
            \\
            1 - \mathcal{O}(K^{-b}) & \kappa = 1
        \end{cases}
        \hspace{3mm} \geq 1 - \eta ,
    \end{align}
    for sufficiently large $K$ where
    \begin{itemize}[leftmargin=*]
        \item $(a)$ follows since $d_k^{i_1} \leq \bar{A}_{\rm max},\, \forall k \in [\![K]\!]$ and $\sum_{k=1}^{K} y_k \leq K ( \bar{A}_{\rm max} + \phi_K ),$ and exploiting $\fy \in \mathscr{F}_{i_1}.$
        \item $(b)$ uses $e^{K \theta_T \phi_K} \to 1,$ as $K$ tends to infinity where the exponent reads
        \begin{align}
            \label{Eq.K_Theta_Phi_Converse}
            K \theta_T \phi_K = 
            \begin{cases}
               K^{(1 - \tau (2\kappa + 3l + b)) / 2} & 0 < \kappa < 1 ,
               \\
               K^{(- (1 + l + b)) / 2} & \kappa = 1 ,
            \end{cases}
        \end{align}
        which given $\tau \geq (\kappa + (1 - \delta_{\kappa,1})l)^{-1}$ vanishes.
        \item $(c)$ employs the Bernoulli's inequality \cite[Ch.~3]{Mitrinovic13} $$(1 - x)^s \geq 1 - sx \; , \; \, 0 \leq \forall x \leq 1 \; , \forall s \geq 1 .$$
        \item $(d)$ exploits the following definition: $f(x) = e^{-tx}(1-x)^t$ with $x = K \theta_T$ and $t = \bar{A}_{\rm max} + \phi_K.$
        \item $(e)$ holds by the Taylor expansion $f(x) = 1 - 2tx + \mathcal{O}(x^2)$ and the lower bound $f(x) \geq 1 - 3tx$ for sufficiently small values of $x.$
        \item $(f)$ uses $T = K^{\tau}$ with $\tau \geq (\kappa + (1 - \delta_{\kappa,1})l)^{-1}.$
    \end{itemize}
    Thereby, Equation~\eqref{Ineq.Cond_Channel_Diff} can then be written as follows
    \begin{align}
       \label{Ineq.Cond_Channel_Diff2}
       & W^{K} \big( \fy \, \big| \, \fc^{i_1} \big) - W^{K} \big( \fy \, \big| \, \fc^{i_2} \big)
       \nonumber\\&
       \leq W^{K} \big( \fy \, \big| \, \fc^{i_1} \big) \cdot \left[ 1 - e^{- \theta_T \sum_{k=1}^{K} d_k^{i_1}} \cdot \left( 1 - \theta_T \right)^{\sum_{k=1}^{K} y_k} \right]
       \nonumber\\&
       \leq \eta \cdot W^{K} \big( \fy \, \big| \, \fc^{i_1} \big) .
    \end{align}
    Now, combining \cref{Eq.Error_Sum_2,Ineq.Cond_Channel_Diff,Ineq.Cond_Channel_Diff2} yields
    \begin{align}
        & P_{e,1}(i_1) + P_{e,2}(i_2,i_1) \geq 1 - \eta \cdot \hspace{-2mm} \sum_{\Z_{i_1} \cap \mathscr{F}_{i_1}} W^{K} \big( \fy \, \big| \, \fc^{i_1} \big) - \eta
        \nonumber\\&
        \overset{(a)}{\geq} 1 - \eta \cdot \sum_{\mathscr{F}_{i_1}} W^{K} \big( \fy \, \big| \, \fc^{i_1} \big) - \eta \overset{(b)}{\geq} 1 - 2\eta ,
    \end{align}
    where for $(a),$ we replaced $\fy \in \mathscr{F}_{i_1} \cap \Z_{i_1}$ by $\fy \in \mathscr{F}_{i_1}$ to enlarge the domain and for $(b),$ we used $\sum_{\fy \in \mathscr{F}_{i_1}} W^{K} ( \fy \, | \, \fc^{i_1}) \leq 1.$ Clearly, this is a contradiction since the error probabilities tend to zero as $N\rightarrow\infty.$ Thus, the assumption in (\ref{Ineq.Converse_Lem_Complement}) is false. This completes the proof of Lemma~\ref{Lem.Converse}.

\renewcommand{\thesectiondis}[2]{\Alph{section}:}
\section{Derivation of The Upper Bounds on The Identification Capacity}
\label{App.Derivation_Upper_Bounds}
\textbf{\textcolor{mycolor12}{Case 1 ($0 < \kappa < 1):$}} Since the codewords inside $\bar{\bar{\mathbb{C}}}_{0,\text{\tiny conv}}^{\fA}$ all belong to a hyper cube $\mathbb{Q}_{\f0}(T,\bar{A}_{\rm max}),$ it follows that the number of codewords $M,$ is bounded by
\begin{align}
    \label{Ineq.Codebook_Size_UB}
    M & = \frac{\text{Vol}\big(\bigcup_{i=1}^{M}\S_{\bar{\bar{\fc}}_{\fA}^i}(T,r_0)\big)}{\text{Vol}(\S_{\bar{\bar{\fc}}_{\fA}^1}(T,r_0))} \stackrel{(a)}{\leq}
    \frac{2^{-0.599T} \cdot \text{Vol}\big( \bar{\bar{\mathbb{C}}}_{0,\text{\tiny conv}}^{\fA} \big)}{\text{Vol}(\S_{\bar{\bar{\fc}}_{\fA}^1}(T,r_0))}
    \nonumber\\&
    \stackrel{(b)}{\leq} 2^{-0.599T} \cdot \frac{\bar{A}_{\rm max}^T}{\text{Vol}(\S_{\bar{\bar{\fc}}_{\fA}^1}(T,r_0))} ,
\end{align}
where $(a)$ exploits the upper bound on the packing density and $(b)$ follows since $\text{Vol}( \bar{\bar{\mathbb{C}}}_{0,\text{\tiny conv}}^{\fA} ) \leq \text{Vol}(\mathbb{Q}_{\f0}(T,\bar{A}_{\rm max})) = \bar{A}_{\rm max}^T.$ Thereby, taking logarithm of \eqref{Ineq.Codebook_Size_UB}, we obtain
\begin{align}
    \label{Ineq.Log_M_UB}
    \log M \leq T \log \bar{A}_{\rm max} - T \log r_0 - T \log \sqrt{\pi} + \mathcal{O}(T) ,
\end{align}
where the dominant term is of order $T \log T.$ Hence, for obtaining a finite value for the upper bound of the rate, $R,$ the scaling of $M$ is induced to be super-exponentially in the rank, i.e., $2^{(T \log T)R}.$ Hence, by setting $M = 2^{(T \log T)R},$ $r_0 = \lambda_{\text{min}} C_{\rm max} / T^{\kappa(1-l\delta_{\kappa,1}) + 2l + b},$ and $\bar{A}_{\rm max} = \mathcal{O}(T^l),$ we get the following bound on the achievable identification rate,
\begin{align}
    \label{Ineq.Rate_UB}
    R \hspace{-.3mm} \leq \hspace{-.3mm} \frac{1}{T \log T} \Big[ \big( \frac{1}{2} \hspace{-.3mm} + \hspace{-.3mm} ( \kappa( 1 - l\delta_{\kappa,1}) + 3l + b ) \big) \, T \log T \hspace{-.3mm} + \hspace{-.3mm} \mathcal{O}(T) \Big] ,
\end{align}
which tends to $\frac{1}{2} + \kappa + 3l$ as $T \to \infty$ and $b \to 0.$

\textbf{\textcolor{mycolor12}{Case 2 ($\kappa = 1):$}} We employ the volume formula for $\text{Vol}_T(\bar{\mathbb{C}}_0^{\fA})$ provided in Lemma~\ref{Lem.Affine_Transformation}. That is, instead of a hyper cube with growing edge length, we use the formula given in \eqref{Eq.Vol_Zonotope} for the special case of $T=N.$ Since we have only one subset $\mathbb{G}_N$ with $\mathbb{G}_N = [\![N]\!],$ we obtain $\text{Vol}(\bar{\bar{\mathbb{C}}}_{0,\text{\tiny conv}}^{\fA}) = C_{\rm max}^N \cdot (\text{det}(\bar{\bar{\fA}}_{\mathbb{G}_N}))^{1/2}.$ Hence, the new equations for derivation of rate as provided in \cref{Ineq.Codebook_Size_UB,Ineq.Log_M_UB,Ineq.Rate_UB} become as follows:
\begin{align}
    \label{Ineq.Codebook_Size_UB-2}
    M & = \frac{\text{Vol}\big(\bigcup_{i=1}^{M}\S_{\bar{\bar{\fc}}_{\fA}^i}(N,r_0)\big)}{\text{Vol}(\S_{\bar{\bar{\fc}}_{\fA}^1}(N,r_0))} \stackrel{(a)}{\leq} \frac{2^{-0.599T} \cdot \text{Vol}\big( \bar{\bar{\mathbb{C}}}_{0,\text{\tiny conv}}^{\fA} \big)}{\text{Vol}(\S_{\bar{\bar{\fc}}_{\fA}^1}(N,r_0))}
    \nonumber\\&
    \leq 2^{-0.599N} \cdot \frac{C_{\rm max}^N \cdot (\text{det}(\bar{\bar{\fA}}_{\mathbb{G}_N}))^{1/2}}{\text{Vol}(\S_{\bar{\bar{\fc}}_{\fA}^1}(N,r_0))} , \hspace{-2mm}
\end{align}
where $(a)$ exploits the upper bound in packing density. Thereby, $\log M \leq \frac{1}{2} \log \text{det}\big( \bar{\bar{\fA}}_{\mathbb{G}_{N}} \big) + N \log C_{\rm max} - N \log r_0 - N \log \sqrt{\pi} + \mathcal{O}(N).$ Next, in order to determine the scaling behavior of terms in the above terms, we use the same line of argument as provided in the achievability of Theorem~\ref{Th.Affine-Capacity}. That is, we obtain $\log \big( \text{det}\big( \bar{\bar{\fA}}_{\mathbb{G}_N} \big) \big) = o(N \log N).$ Thereby, the dominant term is of order $N \log N.$ Hence, for obtaining a finite value for the upper bound of the rate, $R,$ the scaling of $M$ is induced to be $2^{(N \log N)R}.$ Thus, by setting $M = 2^{(N \log N)R}$ and $r_0 = \lambda_{\text{min}} C_{\rm max} / N^{\kappa(1-l\delta_{\kappa,1}) + 2l + b},$ with $\kappa = 1,$ we obtain the following upper bound on $R$
\begin{align}
    \label{Ineq.Rate_UB-2}
    R \leq \frac{1}{N \log N} \Big[ \big( \frac{1}{2} + ( 1 - l) + 2l + b ) \big) \, N \log N + \mathcal{O}(N) \Big] ,
  \end{align}
which tends to $\frac{3}{2} + l$ as $N \to \infty$ and $b \to 0.$

\renewcommand{\thesectiondis}[2]{\Alph{section}:}
\section{Spectrum of Codebooks for Different Message Transmission and Identification Settings}
\label{App.Codebook_Spectrum}

In the literature, various codebook sizes have been reported for different communication tasks, some of which include:
\begin{itemize}[leftmargin=*]
    \item Identification with randomized encoder over q-ary uniform permutation channels \cite{Sarkar24}: Polynomial size, i.e., any codebook with size of $2^{\epsilon_N N^{q-1}}$ where $\epsilon_N \to 0,$ is identifiable.
    \item Shannon problem \cite{S48}, identification for DMC \cite{Salariseddigh_ICC}, and K-identification for binary symmetric channel \cite{Salariseddigh23_BSC_GC23}: Exponential codebook size, i.e., $M_N(R) = 2^{NR}.$
    \item Identification for affine Poisson channels: Super-exponential in the rank, i.e., $M_N(R) \allowbreak = \allowbreak 2^{(\kappa\tau N^{\kappa\tau}\log N)R}.$
    \item Identification for Gaussian fading channels \cite{Salariseddigh_ITW}, K-identification for Gaussian slow fading channels \cite{Salariseddigh_22_ITW}, Binomial channels \cite{Salariseddigh_Binomial_ISIT} and ISI Poisson channels \cite{Salariseddigh-ICC23}: Super-exponential codebook size, i.e., $M_N(R) = 2^{(N\log N)R}.$
    \item Identification with randomized encoder \cite{AD89}: Double-exponential codebook size, i.e., $M_N(R) = 2^{2^{NR}}.$
    \item Identification for Gaussian channels with noiseless feedback \cite{Wiese22}: Codebook size can be arbitrary large function.
\end{itemize}

\end{document}